\documentclass[%
 aip,
 jmp,%
 amsmath,amssymb,
]{revtex4-1}

\usepackage{amsmath,amssymb}
\usepackage{enumerate}
\usepackage{braket}
\usepackage{bbm}
\usepackage{dsfont}

\usepackage{amsthm}
\newtheorem{theorem}{Theorem}
\newtheorem{lemma}{Lemma}
\newtheorem{proposition}{Proposition}
\newtheorem{corollary}{Corollary}

\theoremstyle{definition}
\newtheorem{definition}{Definition}
\newtheorem{remark}{Remark}

\DeclareMathOperator*{\uwlim}{uw-lim}

\newcommand{\defarrow}{\stackrel{\mathrm{def.}}{\Leftrightarrow}}

\newcommand{\cmplx}{\mathbb{C}}
\newcommand{\realn}{\mathbb{R}}

\newcommand{\natn}{\mathbb{N}}
\newcommand{\natz}{\mathbb{N}_0}
\newcommand{\integer}{\mathbb{Z}}
\newcommand{\Real}{\mathrm{Re}} %real part in roman
\newcommand{\Imag}{\mathrm{Im}} %imaginary part in roman

\newcommand{\calL}{\mathcal{L}}
\newcommand{\cH}{\mathcal{H}}
\newcommand{\Hin}{\mathcal{H}_\mathrm{in}  }

\newcommand{\cK}{\mathcal{K}}
\newcommand{\LH}{\mathcal{L} (\mathcal{H}) }
\newcommand{\LHin}{\mathcal{L} (\mathcal{H}_\mathrm{in}  ) }

\newcommand{\LK}{\mathcal{L}  (\mathcal{K})}

\newcommand{\s}{\mathrm{s}}

\newcommand{\cstar}{$C^\ast$}
\newcommand{\A}{\mathcal{A}}
\newcommand{\B}{\mathcal{B}}
\newcommand{\C}{\mathcal{C}}

\newcommand{\M}{\mathcal{M}}
\newcommand{\N}{\mathcal{N}}
\newcommand{\Min}{\mathcal{M}_{\mathrm{in}}}

\newcommand{\tM}{\widetilde{\mathcal{M}}}
\newcommand{\tN}{\widetilde{\mathcal{N}}}
\newcommand{\tL}{\widetilde{\Lambda}}
\newcommand{\tG}{\widetilde{\Gamma}}
\newcommand{\tK}{\widetilde{\mathcal{K}}}
\newcommand{\tpi}{\widetilde{\pi}}
\newcommand{\tV}{\widetilde{V}}
\newcommand{\tE}{\widetilde{\mathcal{E}}}

\newcommand{\vph}{\varphi}

\newcommand{\Ss}{\mathcal{S}_{\sigma}}

\newcommand{\cpchset}[2]{\mathbf{Ch}^{\mathrm{CP}} (#1 \to #2)}
\newcommand{\cpch}[1]{\mathbf{Ch}^{\mathrm{CP}}  (#1)}
\newcommand{\ncpchset}[2]{\mathbf{Ch}^{\mathrm{CP}}_\sigma (#1 \to #2)}
\newcommand{\ncpch}[1]{\mathbf{Ch}^{\mathrm{CP}}_\sigma  (#1)}

\newcommand{\id}{\mathrm{id}}
\newcommand{\unit}{\mathds{1}}

\newcommand{\cocp}{\preccurlyeq_{\mathrm{CP}}}
\newcommand{\eqcp}{\sim_{\mathrm{CP}}}

\newcommand{\norm}[1]{\lVert #1 \rVert}

\newcommand{\phth}{\varphi_\theta}
\newcommand{\psth}{\psi_\theta}

\newcommand{\E}{\mathcal{E}}
\newcommand{\F}{\mathcal{F}}
\newcommand{\seE}{(\M , \varphi_\theta: \theta \in \Theta )}
\newcommand{\seEz}{(\M_0  , \varphi^{(0)}_\theta : \theta \in \Theta)}
\newcommand{\seF}{(\N ,  \psi_\theta :  \theta \in \Theta )}
\newcommand{\phthz}{\varphi_\theta^{(0)}}
\newcommand{\phz}{\varphi^{(0)}}

\newcommand{\thin}{\theta \in \Theta}

\newcommand{\condi}{\mathbb{E}}

\newcommand{\ltwoTh}{\ell^2 (\Theta)}

\newcommand{\LlTh}{\mathcal{L}(\ell^2(\Theta))}

\newcommand{\ltwo}{\ell^2}

\newcommand{\linf}{\ell^\infty}

\newcommand{\AD}{\mathbb{A}(D_{-1})}

\newcommand{\exper}{\mathbf{Exper}}
\newcommand{\CH}{\mathcal{C}\mathcal{H}}
\newcommand{\CHqc}{\mathcal{C}\mathcal{H}^{\mathrm{QC}}}
\newcommand{\ETh}{\mathcal{E}(\Theta)}
\newcommand{\EcTh}{\mathcal{E}^{\mathrm{classical}}(\Theta)}
\newcommand{\Ecl}{\mathcal{E}^{\mathrm{classical}}}

\newcommand{\fth}{\mathbb{F}(\Theta)}

\newcommand{\nonempty}{\neq \varnothing}

\newcommand{\Lshift}{\Lambda^{\mathrm{BDIT}}}

\newcommand{\Lamp}{\Lambda^{\mathrm{amp}}}
\newcommand{\gbarg}{\Gamma^{\mathrm{B}}}

\begin{document}
\preprint{}
\title[]{%
Directed-completeness of quantum statistical experiments in the randomization order%
}
\author{Yui Kuramochi}
\affiliation{School of Physics and Astronomy, Sun Yat-Sen University (Zhuhai Campus),
Zhuhai 519082, China}
\email{yui.tasuke.kuramochi@gmail.com}
\date{}

\begin{abstract}
A parametrized family of normal states on a von Neumann algebra is called a 
statistical experiment, which generalizes the corresponding concepts in
classical statistics and finite-dimensional quantum systems.
We introduce randomization preorder and equivalence relations 
for statistical experiments with a fixed parameter set
and for normal channels with a fixed input space
by post-processing completely positive channels.
In this paper, we prove that the set of equivalence classes of 
statistical experiments or those of normal channels
is an upper and lower directed-complete partially ordered set
with respect to the randomization order,
i.e.\ any increasing or decreasing net of statistical experiments
or channels has its supremum or infimum in the randomization order.
We also show that if the outcome space of each statistical experiment or channel of 
a randomization-monotone net is commutative,
the outcome space of the supremum or infimum
can also be taken to be commutative.
We consider two examples of homogeneous Markov processes of channels
on infinite-dimensional separable Hilbert spaces, namely
block-diagonalization with irrational translation
and ideal quantum linear amplifier channels, 
and explicitly derive their infima.
Throughout the paper, the concept of channel conjugation
is used to obtain results for decreasing channels from 
those for increasing channels.
\end{abstract}
\pacs{03.67.-a, 02.30.Tb, 42.50.-p}
\keywords{quantum statistical experiment, randomization order, directed-completeness, conjugate channel, ideal quantum linear amplifier}
\maketitle
% 42.50.-p Quantum optics

%\subclass{%
%81P45%Quantum information, communication, networks
%\and
%46L53%Noncommutative probability and statistics
%\and
%62B15%  	Theory of statistical experiments
%\and
%47L90%  	Applications of operator algebras to physics
%}

%%%%%%%%%%%%%%%%%%%%%%%INTRODUCTION%%%%%%%%%%%%%
\section{Introduction} \label{sec:intro}
In the area of classical/quantum information or statistics,
we frequently encounter situations in which 
we have a monotonically increasing or decreasing sequence,
or more generally net, 
of information about the system.
Let us give some examples of such situations.
\begin{enumerate}
\item
Suppose that we perform a measurement on a system, which may be 
classical or quantum, 
and obtain a classical outcome $\omega = (\omega_k)_{k=1}^\infty$ 
which takes a value in the countable product 
$\{ 0,1 \}^\natn .$
If $\omega^{(n)} = (\omega_k)_{k=1}^n$ is the first $n$-digits
of $\omega ,$
the information obtained from $\omega^{(n)}$
is increasing with respect to $n$
and upper bounded by that obtained from $\omega .$
We also expect naturally that the information of $\omega^{(n)}$
converges, in some sense, to that of $\omega $
when $n \to \infty .$
\item
As a more general example, consider a system 
described by a von Neumann algebra $\M $
and let $(\M_i)_{i \in I}$ be a net of von Neumann subalgebras of 
$\M$ such that  
$(\M_i)_{i \in I}$ is monotonically increasing in the set inclusion.
%and $\bigcup_{i \in I} \M_i$ is ultraweakly ($\sigma$-weakly) dense
%$\ast$-subalgebra of $\M.$ 
The information when we can access $\M_i$
is increasing with respect to $i$
and upper bounded by the information 
when we can access the whole space $\M .$
If we further assume that
$\bigcup_{i \in I} \M_i$ is an ultraweakly ($\sigma$-weakly) dense
$\ast$-subalgebra of $\M ,$
then we can expect that the information obtained from $\M_i$
converges, in some sense, to that obtained from $\M .$
%Note that, if we take the corresponding sequence of commutative algebras,
%we may regard the previous example as a spacial case of the present one.
\item
Consider a quantum system corresponding to a Hilbert space $\cH$
that undergoes a continuous and homogeneous quantum Markov process
described by
a quantum dynamical semigroup~\cite{lindblad1976,Gorini1978149}
$(\Lambda_{t \ast})_{t >0} ,$
which is a one-parameter family of completely positive (CP) and trace-preserving
maps defined on the trace-class operators on $\cH$ 
and satisfies the semigroup condition
$\Lambda_{t + s \ast}= \Lambda_{t \ast} \circ \Lambda_{s \ast}$
$(s,t >0) .$
Then the information on the system is decreasing 
with respect to $t .$
If $ \Lambda_{t \ast} (\rho_0) $ converges to a fixed state $\rho_e$
for any initial density operator $\rho_0 ,$
we may expect that the information on the initial state will be completely lost
when $t \to \infty.$
What happens in divergent cases?
Can we still define the \lq\lq{}remaining information\rq\rq{} of the system
even when the system density operator 
$\Lambda_{t \ast} (\rho_0) $
is divergent?
\end{enumerate} 

This paper addresses such monotonically increasing or decreasing 
nets of information in the general von Neumann algebra setting
by identifying each \lq\lq{}information\rq\rq{} of a system
with a randomization equivalence class of operator algebraic statistical 
experiments~\cite{jencovapetz2006,gutajencova2007,kuramochi2017minimal} or normal CP channels.
The main finding is that
any randomization-increasing or decreasing 
net of operator algebraic statistical experiments
(normal channels)
has its supremum or infimum.
Furthermore, the supremum or infimum
is classical if each outcome operator algebra of the net
is commutative.
We remark that decreasing sequences of statistical experiments on
a \emph{fixed} finite-dimensional Hilbert space are considered 
by Matsumoto~\cite{matsumoto2012loss}
and it is shown that such a sequence converges in the 
Le Cam distance topology.
Our approach is different from Ref.~\onlinecite{matsumoto2012loss}
in the point that 
we do not fix outcome von Neumann algebras
and give
order-theoretic characterizations of limits,
namely supremum and infimum.
%while the Le Cam distance topology for
%statistical experiments are considered in Ref.~\onlinecite{matsumoto2012loss}.

Let us outline the contents of the paper.
In Sec.~\ref{sec:prel}, we introduce basic notions of
operator algebraic statistical experiments, quantum channels, 
and order theory needed in the main part.
Among them, the notion of the channel conjugation~\cite{kuramochi2018incomp}
and 
the theorem by Iwamura~\cite{iwamura1944} and 
Markowsky~\cite{markowsky1976chain} 
(Theorem~\ref{theo:markow})
play important roles in the proof of the main result:
the former is used to obtain results for decreasing channels
from those for increasing channels,
while the latter enables us to reduce statements about 
general directed sets to 
the case of more specific transfinite sequences.
In Sec.~\ref{sec:smallness}, we establish that
the set of randomization equivalence classes of 
statistical experiments for a given parameter set is well-defined,
which also leads to a similar statement for normal CP channels.
There we slightly generalize the operator algebraic canonical state
introduced in Ref.~\onlinecite{gutajencova2007} to
the case when each normal state of a statistical experiment is 
not faithful and the parameter set is finite.
The well-definedness for infinite parameter sets 
follows from the finite case by
a compactness argument.
In Sec.~\ref{sec:main}, we prove the main result of this paper:
the set of equivalence classes of statistical experiments (normal channels)
is upper and lower directed-complete 
and the set of classical statistical experiments (quantum-classical channels)
is upper and lower Dedekind-closed
(Theorems~\ref{theo:mainch} and \ref{theo:mainex}).
In Sec.~\ref{sec:markov}, we consider two examples of 
homogeneous Markov processes of normal channels 
in infinite-dimensional separable Hilbert spaces and
explicitly derive the infima of these examples.
In Sec.~\ref{sec:final} we give final remarks related to our results.

%%%%%%%%%%%%%%%%%%%%%%%%%PRELIMINARIES%%%%%%%%%%%%%
\section{Preliminaries} \label{sec:prel}
In this section, we introduce mathematical preliminaries and fix the notation.
For general references, 
we refer to Refs.~\onlinecite{takesakivol1,takesakivol2} for operator algebras,
Ref.~\onlinecite{halmos1960naive} for set theory,
and Ref.~\onlinecite{gierz2003continuous}
for order theory.
\subsection{States and channels on operator algebras} %\label{subsec:oa}
In this paper, every \cstar-algebra $\A$ is assumed to have the unit element 
$\unit_\A .$
For a Hilbert space $\cH ,$ $\LH$ denotes the set of bounded operators on $\cH$
and $\unit_\cH$ the identity operator on $\cH .$
The set of normal states on a von Neumann algebra $\M$ is denoted by
$\Ss (\M) .$
The support projection of a normal state $\vph \in \Ss (\M)$ on a von Neumann algebra
$\M$ is denoted by $\s (\vph) .$

Let $\A$ and $\B$ be \cstar-algebras.
A CP and unit-preserving 
linear map $\Lambda \colon \A \to \B$
is called a channel (in the Heisenberg picture).
The set of channels from $\A$ to $\B$ is denoted by
$\cpchset{\A}{\B} .$
In ths Schr\"odinger picture, an input state $\phi$ on $\B$
is mapped to the outcome state $\phi \circ \Lambda $ on $\A .$
Thus the co-domain $\B$ and the domain $\A$ of a channel 
$\Lambda \in \cpchset{\A}{\B}$ are called the input and outcome
spaces, or algebras, of $\Lambda ,$ respectively.
A channel $\Lambda \in \cpchset{\A}{\B}$ is called faithful if 
$\Lambda (A) = 0$ implies $A =0$ for any positive $A \in \A .$
If $\Lambda$ is a representation
(i.e.\ a unit-preserving $\ast$-homomorphism),
$\Lambda$ is faithful if and only if
$\Lambda$ is injective.
A channel $\Lambda \colon \M \to \N$
between von Neumann algebras $\M$ and $\N$ is called normal
if $\Lambda$ is continuous in the ultraweak topologies of $\M$ and $\N ,$
respectively. 
For von Neumann algebras 
$\M$ and $\N ,$
the set of normal channels between them is denoted by $\ncpchset{\M}{\N} .$
A normal channel $\Lambda \in \ncpchset{\M}{\N}$ is called a quantum-classical (QC)
channel if the outcome space $\M$ of $\Lambda$ is commutative.
We also write 
%$\cpch{\A} := \cpchset{\A}{\A}$
%and 
$\ncpch{\M} := \ncpchset{\M}{\M} .$
The identity channel on a \cstar-algebra $\A$ is written as $\id_\A .$

%%%%%%%%%%%%%Stinespring representation%%%%%%%%%%%%%%%%%
Let $\A$ be a \cstar-algebra, let $\Hin$ be a Hilbert space,
and let $\Lambda \in \cpchset{\A}{\LHin}$ be a channel.
A triple $(\cK , \pi , V)$ is called a Stinespring representation of $\Lambda$
if $\cK$ is a Hilbert space, 
$\pi \colon \A \to \LK$ is a representation,
and $V \colon \Hin \to \cK$ is an isometry 
such that
\[
	\Lambda (A) = V^\ast \pi (A) V
	\quad
	(A \in \A) .
\]
A Stinespring representation 
$(\cK , \pi ,V)$
of $\Lambda$ is minimal if the linear span
of $\pi (\A) V \Hin$ is norm dense in $\cK .$
Any channel has a minimal Stinespring representation unique up to
unitary equivalence (Stinespring's dilation theorem~\cite{1955stinespring}).
If $\Lambda \in \ncpchset{\M}{\LHin}$ is a normal channel 
with a minimal Stinespring representation 
$(\cK , \pi , V) ,$
then $\pi$ is a normal representation. 
We remark that, if we do not require the minimality,
Stinespring representation $(\tK , \tpi , \tV)$ of 
a channel $\Lambda \in \cpchset{\A}{\LHin}$ can always be taken 
so that $\tpi$ is faithful.
Such $(\tK , \tpi , \tV)$ is constructed as follows.
We take a minimal Stinespring representation $(\cK ,\pi , V)$ of $\Lambda$
and a faithful representation $\pi_1 \colon \A \to \calL (\cK_1 )$ 
of $\A$ (e.g.\ the universal representation of $\A$).
We define 
%$(\tK , \tpi , \tV)$ by
\[
	\tK := \cK \oplus \cK_1 ,
	\quad
	\tpi := \pi \oplus \pi_1,
	\quad
	\tV \colon \cH \ni x \mapsto
	Vx \oplus 0 \in \cK \oplus \cK_1 .
\]
Then $(\tK , \tpi , \tV)$ is a Stinespring representation of $\Lambda$
and $\tpi$ is faithful.
If $\Lambda$ is normal, the representations 
$\pi_1$ and $\tpi$ can also be taken to be normal.

%inductive limit of C* algebras%
The following notion of the inductive limit of \cstar-algebras~\cite{takeda1955}
will be used in the construction of a supremum of increasing normal channels
(Lemma~\ref{lemm:inch}).
Let $I$ be a directed set and let $(\A_i)_{i \in I}$ be a net of \cstar-algebras.
Suppose that for each $i,j \in I$ with $i \leq j$ 
there exists a faithful representation 
$\pi_{j \gets i} \colon \A_i \to \A_j $
satisfying the following consistency condition
\[
	\pi_{k \gets j} \circ \pi_{j \gets i}
	=
	\pi_{k \gets i}
	\quad
	\text{if $i \leq j \leq k .$}
\]
Note that, if we put $i=j$ in the above condition, we have 
$\pi_{k \gets i } \circ \pi_{i \gets i} = 
\pi_{k \gets i} , $ which implies
$\pi_{i \gets i} = \id_{\A_i}$ by the faithfulness of $\pi_{k \gets i} .$
Then there exist a \cstar-algebra $\widetilde{\A} $
and a net of faithful representations
$\pi_i \colon \A_i \to \widetilde{\A} $
$(i \in I)$
such that
$\pi_{i} = \pi_{j} \circ \pi_{j \gets i}$
$(i \leq j )$ and
$\widetilde{\A}_0 := \bigcup_{i \in I} \pi_i (\A_i)$
is a norm dense $\ast$-subalgebra of $\widetilde{\A} .$
The algebras $\widetilde{\A}_0$ and $\widetilde{\A}$
are called the algebraic and the \cstar-inductive limits of 
$(\A_i , \pi_{j \gets i})_{i \leq j} ,$
respectively.
The representations $(\pi_i)_{i \in I}$ are called the principal representations.

%weak* compactness
Throughout this paper, variants of the following discussion
will be frequently used.
Let $\A$ be a \cstar-algebra, let $\M$ be a von Neumann algebra,
and let $(\Lambda_{i})_{i \in I}$ be a net of channels in
$\cpchset{\A}{\M} .$
Since the closed ball 
\[
	(\M)_r
	:=
	\set{ A \in \M | \| A \| \leq r  }
\]
is ultraweakly compact for each $r \geq 0 ,$ 
by applying Tychonoff's theorem to $\prod_{A \in \A} (\M)_{\|  A \|} ,$
there exists a subnet $(\Lambda_{i(j)})_{j \in J}$ such that
$(\Lambda_{i(j)} (A) )_{j \in J}$
is ultraweakly convergent for each $A .$
If we define $\Gamma \colon \A \to \M$ by
$\Gamma (A) := \uwlim_{j \in J} \Lambda_{i(j)} (A) ,$
where $\uwlim$ denotes the ultraweak limit,
we can easily show that $\Gamma $ is a CP channel.

\subsection{Randomization relations for statistical experiments and normal channels}
For a von Neumann algebra $\M$ and a set $\Theta \neq \varnothing ,$
$\E = \seE$ is called a statistical 
experiment~\cite{jencovapetz2006,gutajencova2007,kuramochi2017minimal}
if $(\phth)_{\thin}$
is a parametrized family of normal states on $\M .$
The von Neumann algebra $\M$ and the set $\Theta$ 
are called the outcome (or sample) space 
and the parameter set of $\E ,$
respectively.
For each set $\Theta \neq \varnothing, $
the class of statistical experiments with the parameter set $\Theta$
is written as $\exper (\Theta ) .$
Since the class of von Neumann algebras is a proper class, 
so is $\exper (\Theta)  .$

We define the \emph{randomization} (or coarse-graining) preorder and equivalence
relations for statistical 
experiments as follows.~\cite{gutajencova2007,kuramochi2017minimal}
\begin{definition}
\label{defi:randse}
Let $\Theta \neq \varnothing$ be a set 
and let $\E = \seE$ and $\F = \seF$
be statistical experiments with the same parameter set $\Theta .$
\begin{itemize}
\item
$\E \cocp \F$ ($\E$ is a randomization of $\F$)
$: \defarrow$
there exists a channel $\alpha \in \cpchset{\M}{\N}$
such that $\phth = \psth \circ \alpha$
for all $\thin .$
\item
$\E \eqcp \F$ ($\E$ is randomization-equivalent to $\F$)
$: \defarrow$
$\E \cocp \F$ and $\F \cocp \E .$
\end{itemize}
\end{definition}
In the above definition, $\E \cocp \F$
if and only if there exists a \emph{normal} channel
$\alpha \in \ncpchset{\M}{\N}$
such that $\phth = \psth \circ \alpha$
for all $\thin $
(Ref.~\onlinecite{gutajencova2007}, the proof of Lemma~3.12).

For a statistical experiment $\E = \seE  $
and a normal channel
$\alpha \in \ncpchset{\N}{\M} ,$
we define the randomized statistical experiment by
$\alpha_\ast (\E) := (\N , \phth \circ \alpha : \thin) .$
We have $\alpha_\ast (\E) \cocp \E$ by definition.

A statistical experiment $\E = \seE$ is called faithful if 
$\phth (A) = 0$ for all $\thin$ implies $A =0$ for any positive $A \in \M .$
$\E$ is faithful if and only if 
$\bigvee_{\thin} \s (\phth) = \unit_\M .$
Any statistical experiment $\E = \seE$ is randomization-equivalent to 
the faithful statistical experiment
$(\M_P , \phth\rvert_{\M_P} : \thin) ,$
where $P : = \bigvee_{\thin} \s (\phth)$
and for each projection $Q \in \M ,$
$\M_{Q} := Q \M Q .$

We can similarly define the randomization relations 
for channels as follows.
\begin{definition}
\label{defi:randch}
Let $\A , \B ,$ and $\C$ be \cstar-algebras
and let $\Lambda \in \cpchset{\A}{\C}$
and $\Gamma \in \cpchset{\B}{\C}$
be channels with the same input space $\C .$
\begin{itemize}
\item
$\Lambda \cocp \Gamma$
($\Lambda$ is a randomization of $\Gamma$)
$:\defarrow$
there exists a channel $\alpha \in \cpchset{\A}{\B}$ such that
$\Lambda = \Gamma \circ \alpha .$
\item
$\Lambda \eqcp \Gamma$
($\Lambda$ is randomization-equivalent to $\Gamma$)
$:\defarrow$
$\Lambda \cocp \Gamma $ and $\Gamma \cocp \Lambda .$
\end{itemize}
\end{definition}
In Definition~\ref{defi:randch},
if $\A , \B ,$ and $\C$ are von Neumann algebras and $\Lambda $ and $\Gamma$
are normal,
then $\Lambda \cocp \Gamma$ if and only if
there exists a normal channel $\alpha \in \ncpchset{\A}{\B}$
such that $\Lambda = \Gamma \circ \alpha .$ 
In Definitions~\ref{defi:randse} and \ref{defi:randch},
$\cocp$ are binary preorders and  
$\eqcp$ are equivalence relations.

The randomization relations for statistical experiments can be characterized by 
those for normal channels as follows.
For a set $\Theta \neq \varnothing$ we denote by $\ltwoTh$ the Hilbert space of 
square summable complex-valued functions on $\Theta .$
We define a normal channel $\condi_\Theta \in \ncpch{\LlTh}$
by 
\[
	\condi_\Theta 
	(A)
	:=
	\sum_{\thin}
	\braket{\delta_\theta |  A \delta_\theta   }
	\ket{\delta_\theta}
	\bra{\delta_\theta}
	\quad
	(A \in \LlTh) ,
\]
where 
\[
	\delta_\theta (\theta^\prime)
	:=
	\begin{cases}
	1 & \text{if $\theta = \theta^\prime , $} \\
	0 & \text{otherwise,}
	\end{cases}
\]
$\braket{f | g} := \sum_{\thin} \overline{f(\theta)}g(\theta ) $
$(f, g \in \ltwoTh)$ is the inner product,
and $\ket{f}\bra{g} \in \LlTh$ $(f,g \in \ltwo (\Theta))$
is given by
$(\ket{f}\bra{g}) h := \braket{g|h} f$
$(h \in \ltwo (\Theta)) .$
For a statistical experiment $\E = \seE $ we define a normal channel
$\Lambda_\E \in \ncpchset{\M}{\LlTh}$ by
\[
	\Lambda_\E (A)
	:=
	\sum_{\thin}
	\phth (A) 
	\ket{\delta_\theta}
	\bra{\delta_\theta}
	\quad
	(A \in \M) .
\]
For statistical experiments $\E$ and $\F$
with the same parameter set $\Theta ,$
$\E \cocp \F$
(respectively, $\E \eqcp \F$ or $\E = \F$)
if and only if
$\Lambda_\E \cocp \Lambda_\F$
(respectively, $\Lambda_\E \eqcp \Lambda_\F$ or $\Lambda_\E = \Lambda_\F$).
\begin{lemma}
\label{lemm:ech}
Let $\Theta \neq \varnothing$ be a set
and let $\Lambda \in \ncpchset{\M}{\LlTh}$
be a normal channel.
Then $\Lambda = \Lambda_{\E}$ for some statistical experiment
$\E \in \exper (\Theta)$
if and only if $\Lambda \cocp \condi_\Theta .$
\end{lemma}
\begin{proof}
Assume that $\Lambda$ can be written as $\Lambda_\E$
for some statistical experiment $\E = \seE .$
Then from 
$\condi_\Theta ( \ket{\delta_\theta} \bra{\delta_\theta }) 
=  \ket{\delta_\theta} \bra{\delta_\theta } , $
we have
$
\condi_\Theta \circ \Lambda
=
\Lambda  ,
$
which implies $\Lambda \cocp \condi_\Theta .$

Conversely, assume $\Lambda \cocp \condi_\Theta .$
Then there exists a normal channel
$\Gamma \in \ncpchset{\M}{\LlTh}$
such that
$\Lambda = \condi_\Theta \circ \Gamma .$
Then for each $A \in \M ,$
\[
	\Lambda (A)
	=
	\condi_\Theta \circ \Gamma (A)
	=
	\sum_{\thin}
	\braket{\delta_\theta | \Gamma (A) \delta_\theta}
	\ket{\delta_\theta} \bra{\delta_\theta } .
\]
Hence if we define $\phth \in \Ss (\M)$ by
$\phth (A) := \braket{\delta_\theta | \Gamma (A) \delta_\theta}$
$(A \in \M) ,$
the statistical experiment $\E := \seE$
satisfies $\Lambda = \Lambda_\E .$
\end{proof}

Conversely, any normal channel can be regarded as a statistical experiment in the 
following way.
For von Neumann algebras $\M$ and $\Min$ and a normal channel 
$\Lambda \in \ncpchset{\M}{\Min} ,$
we define a statistical experiment $\E_\Lambda$ by
\[
	\E_\Lambda
	:=
	(\M , \vph \circ \Lambda :  \vph \in \Ss (\Min)).
\]
We can easily see that for normal channels $\Lambda$ and $\Gamma$
with the same input space $\Min ,$
$\Lambda \cocp \Gamma$
(respectively, $\Lambda \eqcp \Gamma $)
if and only if
$\E_\Lambda \cocp \E_\Gamma$
(respectively, $\E_\Lambda \eqcp \E_\Gamma $).
The following lemma is immediate from these definitions.
\begin{lemma}
\label{lemm:che}
Let $\Min$ be a von Neumann algebra 
and let $\E = (\M , \psi_\vph : \vph \in \Ss (\Min))$
be a statistical experiment.
Then $\E = \E_\Lambda$ for some normal channel
$\Lambda \in \ncpchset{\M}{\Min}$ if and only if 
$\E \cocp \E_{\id_{\Min}} = (\Min , \vph : \vph \in \Ss (\Min)).$
\end{lemma}

%normal extension%
%atode Kuramochi2018no citation wo tuika seyo.
Let $\A$ be a \cstar-algebra, 
let $\Min$ be a von Neumann algebra,
and let $\Lambda \in \cpchset{\A}{\Min}$
be a channel.
Then $\Lambda$ is uniquely extended to a normal channel
$\overline{\Lambda} \in \ncpchset{\A^{\ast \ast}}{\Min} ,$
where $\A^{\ast \ast}$ is the universal enveloping von Neumann algebra of
$\A .$
The normal channel $\overline{\Lambda}$ is called the normal 
extension~\cite{kuramochi2018incomp}
of $\Lambda .$
We have $\Lambda \cocp \overline{\Lambda}$ by definition.
Furthermore, for any normal channel $\Gamma$ with the same input algebra $\Min ,$
$\Lambda \cocp \Gamma $ if and only if
$\overline{\Lambda} \cocp \Gamma $
(Ref.~\onlinecite{kuramochi2018incomp}, Lemma~7),
i.e.\ $\overline{\Lambda}$ is the least normal channel that upper bounds
$\Lambda$ in the randomization preorder.
If $\A$ is a von Neumann algebra and $\Lambda$ is normal,
we have $\Lambda \eqcp \overline{\Lambda} .$

%

%%%%%%%Conjugate (complementary) channel
Let $\A$ be a \cstar-algebra, let $\Hin$ be a Hilbert space,
and let $\Lambda \in \cpchset{\A}{\LHin}$ be a channel.
For a Stinespring representation $(\cK , \pi , V)$ of $\Lambda ,$
we define the conjugate (or complementary) 
channel~\cite{holevo2007complementary,king2007properties,1751-8121-50-13-135302,kuramochi2018incomp}
of $\Lambda$ associated with $(\cK , \pi , V)$ by the normal channel
$\Lambda^c \in \ncpchset{\pi (\M)^\prime}{\LHin}$
given by
\[
	\Lambda^c (B) := V^\ast B V ,
	\quad
	(B \in \pi (\M)^\prime) ,
\]
where the prime denotes the commutant.
While the definition of the conjugate channel explicitly depends on the choice of 
Stinespring representation,
we can show that
any conjugate channels of $\Lambda$ are mutually randomization-equivalent
(Ref.~\onlinecite{kuramochi2018incomp}, Proposition~2).
If the particular choice of conjugate channel is irrelevant,
we denote by $\Lambda^c$ one of the conjugate channels of $\Lambda .$
\begin{proposition}[Ref.~\onlinecite{kuramochi2018incomp}, Theorem~2 and Corollary~1]
\label{prop:conj}
Let $\Hin$ be a Hilbert space and let 
$\Lambda \in \ncpchset{\M}{\LHin}$
and 
$\Gamma \in \ncpchset{\N}{\LHin}$
be normal channels with the same input space $\LHin .$
Then we have the following.
\begin{enumerate}
\item
$\Lambda \cocp \Gamma$ if and only if $\Gamma^c \cocp \Lambda^c .$
\item
$\Lambda \eqcp (\Lambda^c)^c .$
\end{enumerate}
\end{proposition}
If $\A$ is a \cstar-algebra and $\Lambda \in \cpchset{\A}{\LHin}$ is a channel, 
then the double conjugate channel $(\Lambda^c)^c$ coincides with the normal extension 
of $\Lambda $ up to randomization equivalence.

\subsection{Order theory}
Let $(X , \leq )$ be a partially ordered set (poset).
We introduce some notations and definitions as follows.
\begin{itemize}
\item
For each subset $A \subseteq X , $
%$\uparrow A := \set{x \in X | a \leq x \, (\exists a \in A)}$
%and
$\downarrow A := \set{x \in X | \exists a \in A \text{ s.t.\ }x \leq a } .$
%\item
%For each point $x \in X ,$
%$\uparrow x := \, \uparrow \{ x\}$
%and 
%$\downarrow x := \, \downarrow \{ x \} .$
\item
A subset $A \subseteq X$ is called lower if
$A = \, \downarrow A  .$
\item
$X$ is called upper (respectively, lower)
directed if for each $x,y \in X$ there exists $z \in X$
such that $x \leq z$ and $y \leq z$
(respectively, $x \geq z$ and $y \geq z$).
\item
$X$ is called an upper directed-complete poset (upper dcpo)
if every upper directed subset $D$ of $X$ 
has a supremum (i.e.\ least upper bound) $\sup D$
in $X .$
\item
$X$ is called a lower directed-complete poset (lower dcpo)
if every lower directed subset $D$ of $X$ has an infimum 
(i.e.\ greatest lower bound) $\inf D$
in $X .$
\item
A net $(x_i)_{i \in I}$ on $X$ is called increasing (respectively, decreasing)
if $i \leq j$ implies $x_i \leq x_j$ (respectively, $x_j \leq x_i$)
for any $i,j \in I .$
\item
A subset $A\subseteq X$ is called upper (respectively, lower) 
Dedekind-closed~\cite{10.2307/2033201,mcshane1953order} if, 
whenever an upper (respectively, lower) 
directed subset $D$ of $A$ has a supremum $\sup D \in X $
(respectively, infimum $\inf D \in X$), 
then $\sup D\in A$
(respectively, $\inf D \in A$).
%%A subset $A \subseteq X$ is called upper (respectively, lower) 
%Dedekind-closed~\cite{10.2307/2033201,mcshane1953order}
%if upper (respectively, lower) directed subset $D$ of $A$ 
%has a supremum $x= \sup D$ 
%(respectively, an infimum $x=\inf D$) in $X,$
%then $x\in A .$
\end{itemize}
In the above definition, 
$X$ is an upper (respectively, lower) dcpo 
if and only if any increasing (respectively, decreasing) 
net $(x_i)_{i \in I}$ on $X$ has a supremum (respectively, infimum)
in $X .$

Let $X$ and $Y$ be upper (respectively, lower) dcpos.
A map $f \colon X \to Y$ is called upper (respectively, lower)
Scott-continuous~\cite{gierz2003continuous} 
if $f$ is order-preserving and 
$f(\sup D)  = \sup f (D) $
(respectively, $f(\inf D)  = \inf f(D) $)
for any upper (respectively, lower) directed subset $D \subseteq X .$
If $X$ and $Y$ are upper and lower dcpos and 
$f \colon X \to Y$ is upper and lower Scott-continuous,
$f$ is called bi-Scott-continuous.

We occasionally omit the term \lq\lq{}upper\rq\rq{} 
if there is no confusion in the context;
for example, \lq\lq{}a dcpo\rq\rq{} means \lq\lq{}an upper dcpo.\rq\rq{}

We identify, as usual, the cardinality $|S|$ of a set $S$
with the smallest ordinal $\alpha$ satisfying $|\alpha| = |S|$
(the von Neumann cardinal assignment).
We also understand a transfinite sequence $(x_\alpha)_{\alpha < \alpha_0}$
in a set $S$
to be a net on $S$ 
indexed by ordinals $\alpha$ smaller than $\alpha_0$
and greater than or equal to $0 .$
The following theorem is due to Iwamura and Markowsky.
\begin{theorem}[Ref.~\onlinecite{iwamura1944}; Ref.~\onlinecite{markowsky1976chain}, Theorem~1]
\label{theo:markow}
Let $D$ be an infinite directed set.
Then there exists a transfinite sequence 
$(D_\alpha)_{\alpha < | D | }$
of directed subsets of $D$ satisfying the following conditions.
\begin{enumerate}[(i)]
\item
For each $\alpha < |D| ,$ 
$D_\alpha$ is finite if $\alpha$ is finite
and $|D_\alpha| = |\alpha|$ if $\alpha$ is infinite.
\item
$0 \leq \alpha \leq \beta < |D|$ implies
$D_\alpha \subseteq D_\beta .$
\item
$D = \bigcup_{\alpha < |D|} D_\alpha .$
\end{enumerate}
\end{theorem}

By using Theorem~\ref{theo:markow}
Markowsky showed the following.

\begin{proposition}[Ref.~\onlinecite{markowsky1976chain}, Corollary~1]
\label{prop:dc}
A poset $X$ is a dcpo if and only if
every increasing transfinite sequence 
$(x_\alpha)_{\alpha < \alpha_0}$ in $X$
has a supremum in $X .$
\end{proposition}

By slightly modifying the proof of the above proposition
in Ref.~\onlinecite{markowsky1976chain},
we obtain the following two lemmas.

\begin{lemma}
\label{lemm:scott}
Let $X$ and $Y$ be dcpos and 
let $f \colon X \to Y$ be an order-preserving map.
Then the following conditions are equivalent.
\begin{enumerate}[(i)]
\item
$f$ is Scott-continuous.
\item
$f( \sup_{\alpha < \alpha_0} x_\alpha ) = \sup_{\alpha < \alpha_0} f(x_\alpha)$
for any increasing transfinite sequence $(x_\alpha)_{\alpha < \alpha_0}$
in $X .$
\end{enumerate}
\end{lemma}
\begin{proof}
(i)$\implies$(ii) is obvious.
Assume that (i) is not true.
Then there exists a directed subset $D \subseteq X$ satisfying
$f ( \sup D ) \neq \sup f(D) .$
%and hence $\sup f(D) < f (\sup D ) .$
We can take the cardinality of $D$ to be minimal 
so that 
$f (\sup D^\prime) = \sup f(D^\prime)$
for any directed subset $D^\prime \subseteq X$ with $|D^\prime| < |D| .$
$D$ cannot be finite.
Let $(D_\alpha)_{\alpha < |D|}$ be a transfinite sequence 
of directed subsets of $D$
satisfying the conditions (i)-(iii)
of Theorem~\ref{theo:markow}.
Then if we put $x_\alpha := \sup D_\alpha $ for each $\alpha <|D| ,$
the transfinite sequence $(x_\alpha)_{\alpha < |D|}$
is increasing. 
Furthermore 
\[
	f( \sup_{\alpha < |D|} x_\alpha )= f (\sup D) \neq \sup f(D) 
	= \sup_{\alpha < |D|}  ( \sup  f (D_\alpha) )
	=  \sup_{\alpha < |D|} f (x_\alpha) .
\]
Therefore (ii) does not hold.
\end{proof}

\begin{lemma}
\label{lemm:dcsubset}
Let $X$ be a dcpo and let $Y \subseteq X$ be a subset.
Then the following conditions are equivalent.
\begin{enumerate}[(i)]
\item
$Y$ is a Dedekind-closed subset of $X .$
\item
For every increasing transfinite sequence 
$(x_\alpha)_{\alpha < \alpha_0} $
in $Y ,$
$\sup_{\alpha < \alpha_0} x_\alpha \in Y .$
\end{enumerate}
\end{lemma}
\begin{proof}
(i)$\implies$(ii) is obvious.
Assume that (i) is not true.
Then there exists a directed subset $D \subseteq Y$
satisfying $\sup D \not\in Y .$
We can take the cardinality of $D$ to be minimal
so that 
$\sup D^\prime \in Y $
for any directed subset $D^\prime \subseteq Y$ 
with $|D^\prime| < |D| .$
$D$ should be infinite.
Let $(D_\alpha)_{\alpha < |D|}$ be a transfinite sequence 
of directed subsets of $D$
satisfying the conditions (i)-(iii)
of Theorem~\ref{theo:markow}.
Then $x_\alpha := \sup D_\alpha \in Y$ for each $\alpha <|D|$
and the transfinite sequence $(x_\alpha)_{\alpha < |D|}$
is increasing. 
Furthermore $ \sup_{\alpha < |D|} x_\alpha = \sup D \not\in Y .$
Therefore (ii) does not hold.
\end{proof}

\section{The set of randomization equivalence classes} \label{sec:smallness}
In this section we show that for a given parameter set $\Theta \neq \varnothing ,$
the set $\E (\Theta)$ of randomization equivalence classes of quantum statistical experiments 
is well-defined.
% (Proposition~\ref{prop:smallse}).
The construction of $\E (\Theta)$ is based on the operator algebraic canonical state
introduced in Ref.~\onlinecite{gutajencova2007}.
Here we slightly generalize the construction in Ref.~\onlinecite{gutajencova2007} 
to the case 
when each state is not necessarily faithful
and the parameter set is finite.
%Since the proofs in this section are almost parallel to those in 
%Ref.~\onlinecite{gutajencova2007}, most of them are omitted.

\subsection{Connes' cocycle derivative and minimal sufficiency}
For the construction of the canonical state, we need the following concepts.~\cite{takesakivol2}
\begin{definition}[Modular automorphism group]
\label{defi:modular}
Let $\M$ be a von Neumann algebra and let $\vph \in \Ss (\M)$
be a faithful normal state.
Then the modular automorphism group $(\sigma^\vph_t)_{t \in \realn}$
is a unique ultrastrongly ($\sigma$-strongly) 
continuous one-parameter group of automorphisms on $\M$
satisfying the following modular condition:
for each $A , B \in \M$ there exists a function $F \in \AD$ such that
\[
	F (t) = \vph (\sigma^\vph_t (A) B) ,
	\quad
	F(t+i) = \vph (B \sigma^\vph_t (A)) ,
\]
where $D_{-1} : = \set{z \in \cmplx | 0 <  \Imag \, z  <1} $
and, for a domain $D \subseteq \cmplx ,$
$\mathbb{A} (D)$ denotes the set of bounded complex-valued functions 
analytic in $D$
and continuous in $\overline{D} .$
\end{definition}
\begin{definition}[Connes' cocycle derivative]
\label{defi:connes}
Let $\vph$ be a faithful normal state on a von Neumann algebra $\M .$
For each normal state $\psi \in \Ss (\M) ,$
there exists an ultrastrongly continuous one-parameter family
$(u_t)_{t \in \realn}$ of partial isometries in $\M$ satisfying the following conditions.
\begin{enumerate}[(i)]
\item(Cocycle condition)
$u_{s+t} = u_s \sigma^\vph_s (u_t)$
($s,t \in \realn$).
\item
$u_s u_s^\ast = \s (\psi) ,$
$u_s^\ast u_s = \sigma^\vph_s (\s (\psi))$
$(s \in \realn) .$
\item
For each $A, B \in \M ,$ there exists a function
$F \in \AD$ such that
\[
	F (t) = \psi (u_t \sigma^\vph_t (B) A) ,
	\quad
	F(t+i) = \vph ( A u_t \sigma^\vph_t (B))
	\quad
	(t \in \realn) .
\]
\item
$\sigma^\psi_t (A) = u_t \sigma^\vph (A) u_t^\ast$
($A \in \M_{\s (\psi)} , $ $t \in \realn$).
\end{enumerate}
The condition (iii) conversely characterizes the one-parameter family 
$(u_t)_{t \in \realn} .$
We write $ [D\psi , D \vph]_t := u_t $
and call the family $([D \psi , D \vph]_t)_{t \in \realn}$
the cocycle derivative of $\psi$ relative to $\vph .$
\end{definition}

\begin{theorem}[Ref.~\onlinecite{jencovapetz2006}]
Let $\Theta \neq \varnothing$ be a finite set, 
let $\E = \seE$ be a faithful statistical experiment,
and let $\alpha \in \ncpchset{\N}{\M}$ be a faithful normal channel.
%Suppose that there exists a countable set $\Theta_0 \subseteq \Theta$
%satisfying $\bigvee_{\theta \in \Theta_0} \s (\phth) = \unit_\M $
%and define a faithful normal state 
%$\vph := \sum_{\theta \in \Theta_0} \lambda_\theta \phth$
%for some $\lambda_\theta >0 $ with $\sum_{\theta \in \Theta_0} \lambda_\theta =1 .$
Then $\E \eqcp \alpha_\ast (\E)$ if and only if
$[D\phth , D \vph ]_t = \alpha (  [D \phth \circ \alpha , D \vph \circ \alpha]_t  )$
for all $\thin$ and all $t \in \realn ,$
where $\vph := |\Theta|^{-1} \sum_{\thin} \phth .$
\end{theorem}

\begin{definition}[Refs.~\onlinecite{jencovapetz2006,Luczak2014,kuramochi2017minimal}]
\label{defi:msalg}
Let $\E = \seE$ be a statistical experiment and let $\N \subseteq \M$
be a von Neumann subalgebra of $\M .$
\begin{enumerate}
\item
$\N$ is called a sufficient subalgebra for $\E$ (in the sense of CP channel)
if there exists a normal channel $\alpha \in \ncpchset{\M}{\N}$
such that $\phth = \phth \circ \alpha $
for all $\thin .$
\item
$\N$ is called a minimal sufficient subalgebra 
for $\E$
(in the sense of CP channel)
if $\N$ is a sufficient subalgebra and included in any sufficient subalgebra.
\end{enumerate}
\end{definition}

\begin{definition}[Ref.~\onlinecite{kuramochi2017minimal}]
A statistical experiment $\E = \seE$ is called minimal sufficient 
if
$\phth \circ \alpha = \phth$
for all $\thin$ implies $\alpha = \id_\M$
for any normal channel $\alpha \in \ncpch{\M} .$
\end{definition}
It is known~\cite{kuramochi2017minimal} that any statistical experiment
is randomization-equivalent to a minimal sufficient statistical experiment
unique up to normal isomorphism.

If $\Theta \neq \varnothing$ is finite,
a minimal sufficient statistical experiment randomization-equivalent to
a statistical experiment $\E = \seE$ can be constructed as follows.~\cite{jencovapetz2006,gutajencova2007}
Define $\vph := |\Theta|^{-1} \sum_{\thin} \phth ,$
whose support projection $\s (\vph)$ coincides with $\bigvee_{\thin}\s (\phth) .$
Then by restricting the outcome algebra $\M$ to $\M_{\s (\vph)} ,$
we may assume that $\E$ is faithful.
Let $\M_0 \subseteq \M$ 
be the von Neumann subalgebra generated by 
the Connes' cocycle derivatives:
\[
	\M_0:=
	\{
	[D\phth , D \vph]_t , [D\phth , D \vph]_t^\ast
	\}_{(\theta , t ) \in \Theta \times \realn}^{\prime \prime} .
\]
Then $\M_0$ is a minimal sufficient subalgebra with respect to $\E$
and $\E_0 = \seEz$ is a minimal sufficient statistical experiment
randomization-equivalent to $\E ,$
where $\phthz$ is the restriction of $\phth$ to $\M_0 .$
From this construction we can see that,
if $\Theta$ is finite,
a statistical experiment $\E = \seE$ is minimal sufficient
if and only if $\M$ is generated by the Connes' cocycle derivatives
$\{
	[D\phth , D \vph]_t
	\}_{(\theta , t ) \in \Theta \times \realn} .
$

\subsection{Canonical states}
For the definition of the canonical state,
some notions of $\ast$-monoids are needed.

A $\ast$-monoid (or involutive monoid) $M$
is a monoid (i.e.\ a semigroup with a unit element $e$)
with a map $M\ni g \mapsto g^\ast \in M$ satisfying the following condition:
for any $g,h \in M ,$
\[
	g^{\ast \ast } = g ,
	\quad
	(gh)^{\ast} = h^\ast g^\ast .
\]
Such a map $g \mapsto g^\ast$ is called an involution on $M .$
For $\ast$-monoids $M$ and $N ,$
a map $\pi \colon M \to N$ is called a $\ast$-representation if
$\pi (gh) = \pi(g) \pi (h) $ and
$\pi (g^\ast) = \pi (g)^\ast$
for each $g,h \in M$
and $\pi$ is unit-preserving.
%$\pi (e) = e^\prime ,$
%where $e$ and $e^\prime$ are the unit elements of $M$ and $N ,$
%respectively.

Let $M$ be a $\ast$-monoid.
A complex-valued functional $\omega \colon M \to \cmplx$
is called $\ast$-definite~\cite{lindahl1971}
if $(\omega(g_i^\ast g_j))_{1 \leq i,j \leq n}$
is a positive $n \times n$-matrix for all
integer $n \geq 1$ and all $\{ g_i \}_{i=1}^n \subseteq M .$
A $\ast$-definite functional $\omega$ on $M$ is called a state if
$\omega (e) =1 .$
A triple $(\cK , \pi , \xi)$ is called a Gelfand-Naimark-Segal-representation
(GNS-representation) of a
$\ast$-definite functional $\omega$ on $M $
if $\cK$ is a Hilbert space, $\pi \colon M \to \LK$ is a $\ast$-representation,
$\xi \in \cK$ is a vector such that
$
	\omega (g)
	=
	\braket{\xi|\pi (g) \xi}
$
$(g\in M)$
and the closed linear span of $\pi (M) \xi$ coincides with $\cK .$
GNS-representation of a $\ast$-definite functional on $M$ is, if exists,
unique up to unitary equivalence.

For a set $\Theta \neq \varnothing ,$ 
we denote by $M_\Theta$ the free $\ast$-monoid generated by 
$\Theta \times \realn$ satisfying the following condition: 
for each $(\theta , t) \in \Theta \times \realn ,$
\[
	(\theta , 0) (\theta , t) = (\theta ,t) ,
	\quad
	(\theta , t) (\theta , t)^\ast = (\theta , 0) .
\]
The $\ast$-monoid $M_\Theta$ can be characterized by the following universal property:
if $f \colon \Theta \times \realn \to N$ is a map to a $\ast$-monoid $N$ satisfying  
\[
	f(\theta , 0) f(\theta , t) = f(\theta ,t) ,
	\quad
	f(\theta , t) f(\theta , t)^\ast = f(\theta , 0) 
\]
for all $(\theta , t) \in \Theta \times \realn ,$
then there exists a unique $\ast$-representation 
$\pi \colon M_\Theta \to N$ satisfying 
$\pi ((\theta , t)) = f(\theta ,t) $
for all $(\theta , t) \in \Theta \times \realn .$

Let $\Theta \neq \varnothing$ be a finite set. 
For each statistical experiment $\E = \seE ,$
we define the canonical state $\omega_\E$ on $M_\Theta$ as follows.
We take a minimal sufficient statistical experiment 
$\E_0 = \seEz$ randomization-equivalent to
$\E $ and define 
$\phz := |\Theta|^{-1} \sum_{\thin} \phthz .$
Since $\phz$ is faithful,
by taking a GNS-representation of $\phz, $
we may assume that $\M_0$ is realized in a Hilbert space
$\cH_\E $ and $\phz (A) = \braket{\xi_\E | A \xi_\E}$
$(A \in \M_0)$ for a vector $\xi_\E \in \cH_\E$ that separates
$\M_0$ and $\M_0^\prime .$
Since the map
\[
	u \colon 
	\Theta \times \realn 
	\ni
	(\theta , t) \mapsto 
	[D \phthz , D \phz]_t
	\in
	( \M_0 )_1
\] 
satisfies 
\[
	u(\theta , 0) u(\theta , t) = u(\theta , t) ,
	\quad
	u(\theta , t) u(\theta , t)^\ast = \s (\phthz) = u(\theta , 0)
\]
for all $(\theta , t) \in \Theta \times \realn ,$
there exists a unique $\ast$-representation $\pi_\E \colon M_\Theta \to (\M_0)_1$
such that
$\pi_\E ((\theta , t)) =[D\phthz , D \phz]_t $
for all $(\theta , t) \in \Theta \times \realn .$
%where for a von Neumann algebra $\N $ and $r \geq 0 ,$
%$(\N)_r := \set{A \in \N | \| A \| \leq r}$ denotes the closed ball 
%with radius $r .$
We define the canonical state $\omega_\E \in \linf (M_\Theta)$ on $M_\Theta$
by
\[
	\omega_\E (g)
	:=
	\braket{\xi_\E| \pi_\E (g) \xi_\E} 
	=
	\phz (\pi_\E (g)) 
	\quad
	(g \in M_\Theta) ,
\]
where $\linf (S)$ denotes the set of bounded complex-valued functions on
a set $S .$
By the minimal sufficiency, $\M_0$ is generated by $\pi_\E (M_\Theta),$
and hence $\cH_\E$ coincides with the closed linear span of 
$\pi_\E (M_\Theta) \xi_\E .$
Therefore $(\cH_\E  , \pi_\E  , \xi_\E)$ is a GNS-representation of $\omega_\E .$

The following proposition can be shown almost parallel as in 
Ref.~\onlinecite{gutajencova2007} (Theorem~3.5).
\begin{proposition}
\label{prop:canonical}
Let $\Theta \neq \varnothing$ be a finite set 
and let $\E = \seE$ and $\F = \seF$
be statistical experiments.
Then $\E \eqcp \F$ if and only if $\omega_\E = \omega_\F .$
\end{proposition}

\subsection{Construction of $\E  (\Theta )$}
Proposition~\ref{prop:canonical} assures that we may identify 
the set of randomization equivalence classes of statistical experiments 
with the set of canonical states on $M_\Theta$
if $\Theta$ is finite.
Now we consider general $\Theta .$

For a statistical experiment $\E= \seE$
and a subset $\Xi \subseteq \Theta ,$
we define the restriction of $\E$ to $\Xi$ by
\[
	\E \rvert_\Xi
	:=
	(\M , \phth : \theta \in \Xi) .
\]
For a set $\Theta ,$
we denote the set of finite subsets of $\Theta$ by $\fth ,$
which is directed by the set inclusion $\subseteq .$

\begin{proposition}
\label{prop:finite}
Let $\E = \seE$ and $\F = \seF$ be statistical experiments.
Then $\E \cocp \F$ if and only if 
$\E \rvert_F \cocp \F \rvert_F$
for all $F \in \fth .$ 
\end{proposition}
\begin{proof}
\lq\lq{}Only if\rq\rq{} part is obvious.
Assume $\E \rvert_F \cocp \F \rvert_F$
for all $F \in \fth .$
Then for each $F \in \fth$
there exists a channel $\alpha_F \in \cpchset{\M}{\N}$
such that $\phth = \psth \circ \alpha_F$
for all $\theta \in F .$
By Tychonoff's theorem, 
there exist a subnet $(\alpha_{F(i)})_{i \in I}$
and a channel $\alpha \in \cpchset{\M}{\N}$
such that $\alpha_{F(i)} (A) \xrightarrow{uw} \alpha (A) $
for each $A \in \M ,$
where $\xrightarrow{uw}$ denotes the ultraweak convergence.
Since $\theta \in F(i)$ eventually for each $\thin ,$
we have
\[
	\psth \circ \alpha (A)
	=
	\uwlim_{i \in I , \{ \theta \} \subseteq  F(i) }
	\psth \circ \alpha_{F(i)} (A)
	=
	\phth (A) 
\]
for each $A \in \M$ and each $\thin ,$
where 
we used the normality of $\psth $
in the first equality.
Therefore $\E \cocp \F .$
\end{proof}
Now we define the set $\ETh$ of equivalence classes of statistical experiments
for an arbitrary parameter set $\Theta \nonempty .$ 
If $\Theta$ is finite, we define 
$
	\ETh := \set{\omega_\E \in \linf (M_\Theta) | \E \in \exper (\Theta)}
$
and write $[\E] := \omega_\E$ for each $\E \in \exper (\Theta ) .$
If $\Theta$ is infinite,
we define $\ETh$ as the image of the following class-to-set map:
\[
	\exper (\Theta)
	\ni
	\E
	\longmapsto
	[\E]
	:=
	([\E \rvert_F])_{F \in \fth}
	\in
	\prod_{F \in \fth} \E (F) .
\]
Then, for any $\Theta \nonempty $ and $\E , \F \in \exper (\Theta) ,$
Propositions~\ref{prop:canonical} and \ref{prop:finite}
imply that
$\E \eqcp \F$ if and only if $[\E] = [\F] .$
Furthermore, the map $\exper (\Theta) \ni \E \mapsto [\E] \in \ETh$ is surjective.
Therefore we may regard $\ETh$ as the set of equivalence classes of 
statistical experiments with the parameter set $\Theta .$
For a set $\Theta \nonempty ,$ we define a partial order $\cocp$ on $\ETh$
by 
$[\E] \cocp [\F]$
$: \defarrow$
$\E \cocp \F $
$([\E] , [\F] \in \ETh ) .$
We also define the set of equivalence classes of classical statistical experiments by
\[
	\EcTh := \set{[\E] \in \ETh |  \text{the outcome space of $\E$ is commutative}}.
\]

The above definition of $\ETh$ immediately leads to the definition 
of the set of equivalence classes of normal channels as follows.
For a von Neumann algebra $\Min, $
let $\ncpchset{}{\Min}$ denote the class of normal channels with the input space
$\Min. $ 
For a von Neumann algebra $\Min$
and a normal channel 
$\Lambda \in \ncpchset{}{\Min} ,$ 
we define
$[\Lambda] := [\E_\Lambda] \in \E (\Ss (\Min))$
and 
\[	
	\CH (\Min)
	:= 
	\set{[\Lambda] \in \E (\Ss (\Min))  | 
	 	\Lambda \in \ncpchset{}{\Min}
	}
	=
	\downarrow [\E_{\id_{\Min}}],
\]
where the last equality follows from Lemma~\ref{lemm:che}.
Then the map $\ncpchset{}{\Min} \ni \Lambda \mapsto [\Lambda] \in \CH (\Min)$
is surjective. Furthermore, for each $\Lambda , \Gamma \in \ncpchset{}{\Min} ,$
$\Lambda \eqcp \Gamma$ if and only if $[\Lambda] = [\Gamma] .$
Therefore we may regard $\CH (\Min)$ as the set of equivalence classes of 
normal channels with the input space $\Min .$
We define a partial order $\cocp$ on $\CH (\Min)$ by
$[\Lambda ] \cocp [\Gamma]$
$:\defarrow$
$\Lambda \cocp \Gamma$
$([\Lambda] , [\Gamma] \in \CH (\Min)) .$
We also define the set of equivalence classes of QC channels by
\[
	\CHqc (\Min)
	:=
	\set{
	[\Lambda] \in \CH (\Min)
	|
	\text{the outcome space of $\Lambda$ is commutative}
	} .
\]

\begin{remark}\label{rem:otherorder}
The well-definedness of $\ETh$ established above immediately implies
that 
the set of equivalence classes of statistical experiments 
with respect to any equivalence relation less restrictive than $\eqcp$
is also well-defined.
Examples of such equivalence relations are
those induced by normal positive channels and
statistical morphisms.~\cite{kaniowski2013quantum,Buscemi2012}
Note that normal $n$-positive channels for $n \geq 2$ or normal Schwarz channels 
induce exactly the same equivalence relation as $\eqcp$
(Ref.~\onlinecite{kuramochi2017minimal}, Corollary~1).
\end{remark}

\section{Directed-completeness of statistical experiments and channels} \label{sec:main}
In this section, we prove the following two theorems, 
which are the main results of this paper.

\begin{theorem}
\label{theo:mainch}
Let $\Min$ be a von Neumann algebra acting on a Hilbert space $\Hin. $
\begin{enumerate}
\item
$\CH (\Min)$ is an upper and lower dcpo.
\item
$\CHqc (\Min)$ is an upper and lower Dedekind-closed subset of $\CH (\Min) .$ 
\end{enumerate}
\end{theorem}

\begin{theorem}
\label{theo:mainex}
Let $\Theta \nonempty$ be a set.
\begin{enumerate}
\item
$\ETh$ is an upper and lower dcpo.
\item
$\EcTh$ is an upper and lower Dedekind-closed subset of 
$\ETh .$
\end{enumerate}
\end{theorem}

We split the proof into some lemmas.

We first consider increasing normal channels.
The following two lemmas are essential for 
the construction of a supremum.
\begin{lemma}
\label{lemm:increasing_alg}
Let $\Lambda \in \cpchset{\A}{\LHin}$ be a channel,
let $\Gamma \in \ncpchset{\N}{\LHin}$ be a normal channel,
let $(\A_i)_{i\in I}$ be a net of unital \cstar-subalgebras of
$\A $
such that $\A_{i_1} \subseteq \A_{i_2}$ for each 
$i_1 \leq i_2 $
and 
$\A_0 := \bigcup_{i \in I} \A_i $ is a norm dense $\ast$-subalgebra of $\A ,$
and let $\Lambda_i$ be the restriction of $\Lambda$
to $\A_i .$
Suppose that $\Lambda_i \cocp \Gamma$ for all $i \in I.$
Then $\Lambda \cocp \Gamma .$
\end{lemma}
\begin{proof}
By assumption, for each $i \in I$ 
there exists a channel 
$\Psi_i \in \cpchset{\A_i}{\N}$
such that
$\Lambda_i = \Gamma \circ \Psi_i . $
For each $i \in I$ we define a map 
$\widetilde{\Psi}_i \colon \A_0 \to \N$
by
\[
	\widetilde{\Psi}_i (A)
	:=
	\begin{cases}
	\Psi_i(A) 
	&\text{if $A \in \A_i ;$}
	\\
	0
	&
	\text{otherwise.}
	\end{cases}
\]
By Tychonoff's theorem,
we can take a subnet 
$(\widetilde{\Psi}_{i(j)})_{j \in J}$
such that the ultraweak limit
\[
	\Psi_{0} (A) := \uwlim_{j \in J}  \widetilde{\Psi}_{i(j) } (A)
	\in
	(\N)_{\norm{A}}
\]
exists for each $A \in \A_0 .$
Since
$\widetilde{\Psi}_{i(j)} (c_1A +c_2B) =  
c_1 \widetilde{\Psi}_{i(j)} (A)
+
c_2 \widetilde{\Psi}_{i(j)} (B)  $ 
eventually for each $A, B \in \A_0$ and each $c_1 ,c_2 \in \cmplx, $
$\Psi_{0}$ is a bounded linear map.
From the complete positivity of $\Psi_i$ 
$(i \in I) ,$
we can show that $\Psi_{0} $ is CP 
in the following sense:
for each $n \geq 1$ and each
$(A_k)_{k=1}^n \subseteq \A_{0} ,$
the $n \times n$ matrix
$(\Psi_{0}  (A_i^\ast A_j)  )_{i, j =1}^n$
is positive.
Hence $\Psi_{0}$
uniquely extends to a CP channel
$\Psi \in \cpchset{\A}{\N } .$
Then for each $i \in I$
and each $A \in \A_i ,$
\[
	\Gamma \circ \Psi (A)
	=
	\uwlim_{j \in J }
	\Gamma (
	\widetilde{\Psi}_{i(j)} (A)
	)
	=
	\uwlim_{j \in J , i(j) \geq i}
	\Gamma (
	\Psi_{i(j)} (A)
	)
	=
	\uwlim_{j \in J , i(j) \geq i}
	\Lambda_{i(j)} (A)
	=
	\Lambda (A) ,
\]
where we used the normality of
$\Gamma$
in the first equality.
Since $\A_{0}$ is norm dense in $\A ,$
this implies 
$\Lambda = \Gamma \circ \Psi \cocp \Gamma . $
\end{proof}

\begin{lemma}
\label{lemm:embed}
Let $\Lambda_1 \in \cpchset{\A_1}{\LHin}$ be a channel
and let $\Lambda_2 \in \ncpchset{\M_2}{\LHin}$
be a normal channel.
Assume $\Lambda_1 \cocp \Lambda_2 .$
Then there exist a von Neumann algebra $\tM_2 ,$
a faithful representation 
$\pi_1 \colon \A_1 \to \tM_2 ,$
and a normal channel 
$\tL_2 \in \ncpchset{\tM_2}{\LHin}$
such that
$\tL_2 \eqcp \Lambda_2$
and
$\Lambda_1 = \tL_2 \circ \pi_1 .$
If both $\A_1$ and $\M_2$ are commutative,
$\tM_2$ can be taken to be commutative.
\end{lemma}
%Roughly speaking, Lemma~\ref{lemm:embed}
%says that if we take the outcome space of $\tL_2$ 
%\lq\lq{}sufficiently large\rq\rq{}, 
%$\Lambda_1$ can be realized as the restriction of $\tL_2$ to an appropriate subalgebra.
\begin{proof}
By the assumption $\Lambda_1 \cocp \Lambda_2$
there exists a channel $\Phi \in \cpchset{\A_1}{\M_2}$
satisfying $\Lambda_1 = \Lambda_2 \circ \Phi .$
Suppose that $\M_2$ acts on a Hilbert space $\cH_2 $
and let $(\cK_1 , \pi_1 , V_1)$ be a Stinespring representation 
of $\Phi$ such that $\pi_1 \colon \A_1 \to \calL (\cK_1)$ is faithful.
We define a (possibly non-unital) normal $\ast$-homomorphism
$\rho \colon \M_2 \to \calL (\cK_1)$
by 
$\rho (B) := V_1 B V_1^\ast $
$(B \in \M_2) .$
Then for each $n \geq 1 , $
each $(A_i)_{i =1}^{n+1} \subseteq \A_1 ,$
and each 
$(B_i)_{i=1}^n \subseteq \M_2$ we have
\[
	V_1^\ast \pi_1 (A_1) V_1 
	=
	\Phi (A_1) \in \M_2 ,
\]
\begin{align*}
	&V_1^\ast 
	\pi_1 (A_1)
	\rho (B_1)
	\cdots
	\pi_1 (A_n)
	\rho (B_n)
	\pi_1 (A_{n+1})
	V_1
	\\
	&=
	V_1^\ast \pi_1 (A_1) V_1 B_1 V_1^\ast
	\cdots
	V_1^\ast \pi_1 (A_n) V_1 B_n V_1^\ast
	\pi_1 (A_{n+1}) V_1
	\\
	&=
	\Phi (A_1) B_1 \cdots \Phi (A_n) B_n \Phi (A_{n+1})
	\\
	& \in \M_2 .
\end{align*}
This implies $V_1^\ast C V_1 \in \M_2$ for each
$C \in \tM_2 ,$
where $\tM_2 \subseteq \calL (\cK_1)$ is the von Neumann algebra
generated by $\pi_1 (\A_1) \cup \rho (\M_2 ) .$ 
Hence we may define a normal channel 
$\tL_2 \in \ncpchset{\tM_2}{\LHin}$
%and 
%$\Psi \in \ncpchset{\M_2}{\tM_2}$
by
\begin{gather*}
	\tL_2 (C)
	:=
	\Lambda_2 ( V_1^\ast C V_1 )
	\quad
	(C \in \tM_2) .
%	\\
%	\Psi (B)
%	:=
%	\rho (B) 
%	+ \vph (B)  (  \unit_{\cK_1} - V_1 V_1^\ast )
%	\quad
%	(B \in \M_2) ,
\end{gather*}
We have $\tL_2 \cocp \Lambda_2$ by definition.
On the other hand, if we define
$\Psi \in \ncpchset{\M_2}{\tM_2}$
by
\[
\Psi (B)
	:=
	\rho (B) 
	+ \vph (B)  (  \unit_{\cK_1} - V_1 V_1^\ast )
	\quad
	(B \in \M_2) ,
\]
where $\vph \in \Ss (\M_2)$ is a fixed normal state,
then for each $B \in \M_2 $
\[
	\tL_2 \circ \Psi (B)
	=
	\Lambda_2 \left(
	V_1^\ast V_1 B V_1^\ast V_1
	+
	\vph (B)
	 V_1^\ast    (\unit_{\cK_1} - V_1V_1^\ast)    V_1
	\right)
	=
	\Lambda_2 (B) .
\]
Hence $\tL_2 \eqcp \Lambda_2 .$
For each $A \in \A_1 ,$ we have
\[
	\Lambda_1 (A)
	=
	\Lambda_2 \circ \Phi (A)
	=
	\Lambda_2
	(V_1^\ast \pi_1 (A) V_1)
	=
	\tL_2 \circ \pi_1 ( A) .
\]
Therefore $(\tM_2 , \pi_1 , \tL_2)$ satisfies all the conditions 
of the claim.

Next, we consider the case where $\A_1$ and $\M_2$ are commutative.
Let $\Phi \in \cpchset{\A_1}{\M_2}$ be the same as the previous paragraph.
Then by the commutativity there exists a channel 
$\widetilde{\Phi} \in \cpchset{\A_1 \otimes \M_2}{ \M_2}$
such that
$
\widetilde{\Phi} (A \otimes B)
= 
\Phi (A) B
$
$(A \in \A_1 , B \in \M_2) ,$
where $\A_1 \otimes \M_2$ denotes 
the injective \cstar-tensor product.
We define $\Gamma_2 \in \cpchset{\A_1 \otimes \M_2}{\LHin}$
by 
$\Gamma_2 := \Lambda_2 \circ \widetilde{\Phi} ,$
$\N_2$ by 
the universal enveloping von Neumann algebra
$(\A_1 \otimes \M_2)^{\ast \ast} ,$
and
$\tG_2 \in \ncpchset{\N_2}{\LHin}$
by the normal extension of $\Gamma_2 .$
Since $\A_1 \otimes \M_2$ is commutative,
so is $\N_2 .$
By definition we have $\Gamma_2 \cocp \Lambda_2 ,$
and hence $\tG_2 \cocp \Lambda_2 .$
On the other hand, for each $B \in \M_2$
we have
$\Lambda_2 (B) = \Lambda_2 (  \widetilde{\Phi} (\unit_{\A_1} \otimes B  )   )
= \Gamma_2  (\unit_{\A_1} \otimes B) ,
$
which implies $\Lambda_2 \cocp \Gamma_2 \cocp \tG_2 .$
Therefore $\Lambda_2 \eqcp \tG_2 .$
If we define a representation 
$\pi \colon \A_1 \ni A \mapsto A \otimes \unit_{\M_2} \in \N_2  ,$
then for each $A \in \A_1$ we have
\[
	\tG_2 \circ \pi (A)
	=
	\Gamma_2 (A \otimes \unit_{\M_2})
	=
	\Lambda_2 (\Phi (A))
	=
	\Lambda_1 (A) .
\]
Therefore 
$(\tM_2 , \pi_1 , \tL_2) = (\N_2 , \pi , \tG_2)$
satisfies all the conditions of the claim.
\end{proof}

\begin{lemma}
\label{lemm:inch}
$\CH (\LHin)$ is an upper dcpo.
\end{lemma}
\begin{proof}
By Proposition~\ref{prop:dc},
we have only to establish the existence of a supremum of
an arbitrary increasing transfinite sequence 
$([\Lambda_\alpha ])_{ \alpha < \alpha_0}$
in $\CH (\LHin) .$
Let $\M_\alpha$ be the outcome space of $\Lambda_\alpha .$
We inductively construct a transfinite sequence
$
(
\widetilde{\M}_\alpha , 
(\pi_{\alpha \gets \beta} )_{\beta \leq \alpha}, 
\widetilde{\Lambda}_\alpha
)_{\alpha  < \alpha_0}
$
such that for each $0 \leq \alpha \leq \beta \leq \gamma < \alpha_0 ,$
\begin{itemize}
\item[$\bullet$]
$\widetilde{\M}_\alpha$
is a von Neumann algebra;
\item[$\bullet$]
$\pi_{\beta \gets \alpha} \colon \widetilde{\M}_\alpha \to \widetilde{\M}_\beta$
is a (not necessarily normal) faithful representation
satisfying
$\pi_{\alpha \gets \alpha} = \id_{\tM_\alpha} $
and
$
\pi_{\gamma  \gets \alpha} 
= 
\pi_{\gamma \gets \beta} 
\circ 
\pi_{\beta \gets \alpha}   ;$
\item[$\bullet$]
$\tL_\alpha \in \ncpchset{\tM_\alpha}{\LHin}$
is a normal channel satisfying 
$\tL_\alpha \eqcp \Lambda_\alpha  $
and
$\tL_\alpha = \tL_\beta \circ \pi_{\beta \gets \alpha} .$
\end{itemize}
We define $\tM_0 := \M_0,$
$\pi_{0\gets 0} := \id_{\M_0} ,$
and 
$\tL_0 := \Lambda_0 .$
Now for an ordinal $0 < \gamma < \alpha_0$
suppose that we have constructed 
$
(
\widetilde{\M}_\alpha , 
(\pi_{\alpha \gets \beta} )_{\beta \leq \alpha}, 
\widetilde{\Lambda}_\alpha
)_{\alpha < \gamma}
$
satisfying the required properties.
Let $\A_{0\gamma }$ and $\A_\gamma$
be the algebraic and the \cstar-inductive limits 
of
$(\tM_\alpha , \pi_{\beta \gets \alpha})_{\alpha \leq \beta < \gamma} ,$
respectively,
and let $\sigma_{\gamma \gets \alpha} \colon \tM_\alpha \to \A_\gamma$
be the principal isomorphism such that
$\sigma_{\gamma \gets \alpha} 
= \sigma_{\gamma \gets \beta} 
\circ 
\pi_{\beta \gets \alpha}
$
$(\forall \alpha \leq \forall  \beta < \gamma) .$
By identifying $\tM_\alpha$ with $\sigma_{\gamma \gets \alpha} (\tM_\alpha) ,$
we may regard $(\tM_\alpha)_{\alpha < \gamma}$
as a monotonically increasing transfinite sequence of 
\cstar-subalgebras of $\A_{\gamma} $ 
such that
$\A_{0\gamma} = \bigcup_{\alpha < \gamma } \tM_\alpha   .$
From the condition 
$\tL_\alpha = \tL_\beta \circ \pi_{\beta \gets \alpha}$
$(\forall \alpha \leq \forall \beta < \gamma) ,$
we may define a bounded linear map
$\Phi_{0\gamma} \colon \A_{0\gamma} \to \LHin$
by
$
	\Phi_{0 \gamma} (A)
	:=
	\tL_\alpha (A) 
$
if $A \in \tM_\alpha ,$
which is well-defined irrespective of the choice of  $\alpha .$
The linear map $\Phi_{0\gamma}$ uniquely extends to a CP channel
$\Phi_\gamma \in \cpchset{\A_\gamma}{\LHin} .$
Since $\tL_\alpha \cocp \Lambda_\gamma$ for all $\alpha < \gamma ,$
Lemma~\ref{lemm:increasing_alg} implies
$\Phi_\gamma \cocp \Lambda_\gamma . $
Therefore by Lemma~\ref{lemm:embed}
there exist a von Neumann algebra $\tM_\gamma ,$
faithful representation 
$\rho_\gamma \colon \A_\gamma \to \tM_\gamma ,$
and a normal channel 
$\tL_\gamma \in \ncpchset{\tM_\gamma}{\LHin}$
such that 
$\Lambda_\gamma \eqcp \tL_\gamma $
and 
$\Phi_\gamma = \tL_\gamma \circ \rho_\gamma .$
We define $\pi_{\gamma \gets \gamma} := \id_{\tM_\gamma}$ 
and
for each $\alpha < \gamma$ define a faithful representation
$\pi_{\gamma \gets \alpha} \colon \tM_\alpha \to \tM_\gamma$
by
$\pi_{\gamma \gets \alpha} := \rho_\gamma \circ \sigma_{\gamma \gets \alpha} .$
Then it is straightforward to show that
$(\tM_\alpha , (\pi_{\alpha \gets \beta} )_{\beta \leq \alpha}, 
\tL_\alpha)_{\alpha \leq \gamma }$
satisfies the required properties.
Thus by induction we have constructed 
$(\tM_\alpha , (\pi_{\alpha \gets \beta} )_{\beta \leq \alpha}, \tL_\alpha)_{\alpha < \alpha_0 } .$

Now define $\A_{\alpha_0}$ by 
the \cstar-inductive limit of 
$(\tM_\alpha , \pi_{\beta \gets \alpha})_{\alpha \leq \beta < \alpha_0} $
and let $\sigma_{\alpha_0 \gets \alpha} \colon \tM_\alpha \to \A_{\alpha_0}$
be the principal isomorphism such that
$\sigma_{\alpha_0 \gets \alpha} 
= \sigma_{\alpha_0 \gets \beta} 
\circ 
\pi_{\beta \gets \alpha}
$
$(\forall \alpha \leq \forall  \beta < \alpha_0) .$
Then there exists a channel 
$\Phi_{\alpha_0} \in \cpchset{\A_{\alpha_0}}{\LHin} $
satisfying
$\tL_\alpha = \Phi_{\alpha_0} \circ \sigma_{\alpha_0 \gets \alpha } $
for all $\alpha < \alpha_0 .$
Then, for any normal channel $\Gamma \in \ncpchset{\N}{\LHin}$
satisfying 
$\Lambda_\alpha \cocp \Gamma$
for all $\alpha < \alpha_0 ,$
we have $\Phi_{\alpha_0} \cocp \Gamma$
by Lemma~\ref{lemm:increasing_alg}.
%similarly as in the proof of $\Phi_\gamma \cocp \Lambda_\gamma .$
Thus if we define $\tL_{\alpha_0} \in \ncpchset{\A_{\alpha_0}^{\ast\ast}}{\LHin}$
as the normal extension of $\Phi_{\alpha_0} ,$
we have $\tL_{\alpha_0} \cocp \Gamma.$
Since $\Lambda_\alpha \cocp \tL_{\alpha_0} $ 
$(\alpha < \alpha_0)$
is immediate from the definition, 
$[\tL_{\alpha_0}]$ is 
a supremum of 
$([\Lambda_\alpha ])_{\alpha < \alpha_0} .$
\end{proof}

\begin{lemma}
\label{lemm:increasingcomm}
$\CHqc (\LHin)$ is an upper Dedekind-closed subset of $\CH (\LHin) .$ 
\end{lemma}
\begin{proof}
By Lemma~\ref{lemm:dcsubset}, we have only to show 
that $[\tL_{\alpha_0}] : =\sup_{\alpha < \alpha_0} [\Lambda_\alpha] \in \CHqc (\LHin)$
for any increasing transfinite sequence 
$([\Lambda_\alpha])_{\alpha < \alpha_0}$
in $\CHqc (\LHin) .$
We apply the construction of $(\tL_\alpha)_{\alpha < \alpha_0}$
given in Lemma~\ref{lemm:inch}. 
Then we can construct the outcome space $\tM_\alpha$ of $\tL_{\alpha}$
to be commutative for each $\alpha < \alpha_0 .$
In this case, the outcome space $\A_{\alpha_0}^{\ast \ast}$
of the supremum $\tL_{\alpha_0} $ is also
commutative.
Therefore $[\tL_{\alpha_0}] \in \CHqc (\LHin) .$
\end{proof}

We next consider decreasing channels.

\begin{lemma}
\label{lemm:decreasingch}
$\CH (\LHin)$ is a lower dcpo. 
\end{lemma}
\begin{proof}
Let $([\Lambda_i])_{i \in I}$ be a decreasing net on
$\CH (\LHin) .$
Then Proposition~\ref{prop:conj} implies that the net $([\Lambda_i^c])_{i \in I}$
of conjugate channels 
is increasing 
and hence by Lemma~\ref{lemm:inch}
has a supremum $[\Gamma] \in \CH (\LHin) .$
Again by Proposition~\ref{prop:conj},
$[\Gamma^c]$ is an infimum of
$([(\Lambda_i^c)^c])_{i \in I} =([\Lambda_i])_{i \in I} ,$
which proves the claim.
\end{proof}

To establish the lower Dedekind-closedness of 
$\CHqc (\LHin) ,$
we use an operator algebraic version of the quantum no-broadcasting theorem.~\cite{kuramochi2018access}
A normal channel $\Lambda \in \ncpchset{\M}{\Min}$ is said to be broadcastable
if there exists a channel
$\Psi \in \cpchset{\M \otimes \M}{\M}$ such that
$\Lambda (A) = \Lambda \circ \Psi (A \otimes \unit_\M)
= \Lambda \circ \Psi (\unit_\M \otimes A )  $
for all $A \in \M .$
Here $\M \otimes \M$ denotes the injective \cstar-tensor product.
Such a channel $\Psi$ is called a broadcasting channel of $\Lambda .$

\begin{lemma}[No-broadcasting theorem for normal channel]
\label{lemm:nobroadcast}
Let $\Lambda \in \ncpchset{\M}{\Min}$ be a normal channel.
Then $\Lambda$ is broadcastable if and only if 
$\Lambda$ is randomization-equivalent to a QC channel.
\end{lemma}
\begin{proof}
Application of Corollary~1 of Ref.~\onlinecite{kuramochi2018access}
to $\E_\Lambda .$
\end{proof}

\begin{lemma}
\label{lemm:decreasingqc}
$\CHqc (\LHin)$ is a lower Dedekind-closed subset of 
$\CH (\LHin).$ 
\end{lemma}
\begin{proof}
By Lemma~\ref{lemm:dcsubset}, we have only to show 
$\inf_{\alpha < \alpha_0} [\Lambda_\alpha] \in \CHqc (\LHin)$
for any decreasing transfinite sequence 
$([\Lambda_\alpha])_{\alpha < \alpha_0}$
in $\CHqc (\LHin) .$
We apply the same construction of 
a supremum channel
$\tG \in \ncpchset{\tN}{\LHin}$
of increasing 
$(\Lambda_\alpha^c)_{\alpha < \alpha_0} $
as in the proof of Lemma~\ref{lemm:inch}.
Then there exists an increasing transfinite sequence
$(\A_\alpha)_{\alpha < \alpha_0}$
of \cstar-subalgebras of $\tN$
such that 
the restriction
$\Gamma_\alpha := \tG \rvert_{\A_\alpha}$
is randomization-equivalent to $\Lambda_\alpha^c$
for each $\alpha < \alpha_0 $
and 
$\bigcup_{\alpha < \alpha_0} \A_\alpha$
is ultraweakly dense in $\tN .$
Take a Stinespring representation 
$(\tK , \tpi , \tV)$
of 
$\tG$
such that $\tpi$ is faithful and normal,
and define
$\tM_\alpha := \tpi (\A_\alpha)^\prime  $
and
$\tM := \tpi (\tN)^\prime .$
Then $\tM = \bigcap_{\alpha < \alpha_0} \tM_\alpha .$
We also define
$\tL_\alpha \in \ncpchset{\tM_\alpha}{\LHin}$
$(\alpha < \alpha_0)$
and
$\tL \in \ncpchset{\tM}{\LHin}$
by
\begin{gather*}
	\tL_\alpha (A)
	:=
	\tV^\ast A \tV 
	\quad
	(A \in \tM_\alpha) ,
	\\
	\tL (A)
	:=
	\tV^\ast A \tV 
	\quad
	(A \in \tM) . 
\end{gather*}
Since the double conjugate channel of a normal channel
is randomization-equivalent to the original channel,
we have
$\Lambda_\alpha \eqcp \tL_\alpha$
for each $\alpha< \alpha_0$
and 
$ [\tL] = \inf_{\alpha < \alpha_0}
[\Lambda_\alpha] .$

By assumption and Lemma~\ref{lemm:nobroadcast},
we can take a broadcasting channel
$T_\alpha \in \cpchset{\tM_\alpha \otimes \tM_\alpha}{\tM_\alpha}$
of $\tL_\alpha$
for each $\alpha < \alpha_0 .$
Since $\tM \otimes \tM \subseteq \tM_\alpha \otimes \tM_\alpha$
(e.g.\ Ref.~\onlinecite{takesakivol1}, Proposition~4.22; 
or Ref.~\onlinecite{brown2008c}, Proposition~3.6.1),
we may define $S_\alpha \in \cpchset{\tM \otimes \tM}{\tM_\alpha}$
as the restriction of $T_\alpha$ to
$\tM \otimes \tM .$
By Tychonoff's theorem,
there exist a subnet $(S_{\alpha(j)})_{j \in J}$ 
and
a channel
$S \in \cpchset{\tM \otimes \tM}{\calL (\tK)}$
such that
$S(X) = \uwlim_{j \in J} S_{\alpha (j)} (X)$
$(X \in \tM \otimes \tM ) .$
Since $S_{\alpha (j)} (X) \in \tM_\alpha$
eventually for each $X\in \tM \otimes \tM $
and each $\alpha < \alpha_0 ,$
the ultraweak closedness of each $\M_\alpha$ implies
$S (X) \in \bigcap_{\alpha < \alpha_0} \tM_\alpha = \tM $
for all $X \in \tM \otimes \tM .$
Thus 
$S \in \cpchset{\tM \otimes \tM}{\tM } .$
Then for each $A \in \tM$ we have
\begin{align*}
	\tL \circ S(A \otimes \unit_{\tK})
	&=
	\tV^\ast 
	\left(
	\uwlim_{j \in J} S_{\alpha (j)}
	(A \otimes \unit_{\tK})
	\right)
	\tV
	\\
	&=
	\uwlim_{j \in J}
	\tV^\ast 
	S_{\alpha (j)} (A \otimes \unit_{\tK})
	\tV
	\\
	&=
	\tV^\ast A \tV ,
	\\
	\tL \circ S (  \unit_{\tK}\otimes A)
	&=
	\uwlim_{j \in J}
	\tV^\ast 
	S_{\alpha (j)}  (  \unit_{\tK}\otimes A)
	\tV
	\\
	&=
	\tV^\ast A \tV  
	.
\end{align*}
Therefore $S$ is a broadcasting channel of $\tL ,$
and hence Lemma~\ref{lemm:nobroadcast} implies  
$[\tL] \in \CHqc (\LHin) .$
\end{proof}

\noindent
\textit{Proof of Theorem~\ref{theo:mainch}.}
If $\Min = \LHin ,$ the claim is immediate from Lemmas~\ref{lemm:inch},
\ref{lemm:increasingcomm}, \ref{lemm:decreasingch}, and \ref{lemm:decreasingqc}.
Since $\CH (\Min)$ and $\CHqc (\Min)$ can be identified with 
the lower subsets
\begin{gather*}
	\set{[\Lambda] \in \CH (\LHin)  |  [\Lambda] \cocp [\id_{\Min}]  }
	\\
	\set{[\Lambda] \in \CHqc (\LHin)  |  [\Lambda] \cocp [\id_{\Min}]  }
\end{gather*}
of $\CH (\LHin) ,$ respectively,
the claim for general $\Min$ follows from that for $\LHin .$
\qed

\noindent
\textit{Proof of Theorem~\ref{theo:mainex}.}
Let $([\E_i])_{i \in I}$ be an increasing (respectively, decreasing)
net in $\ETh .$
Then the net $([\Lambda_{\E_i}])_{i \in I}$ of normal channels 
is increasing (respectively, decreasing) in $\CH (\calL (\ltwo (\Theta)))$
and hence has a supremum (respectively, infimum)
$[\tL] \in \CH (\calL (\ltwo (\Theta))) $
by Theorem~\ref{theo:mainch}~(i).
If $([\E_i])_{i \in I}$ is increasing,
we have $\Lambda_{\E_i} \cocp \condi_{\Theta}$
for all $i \in I$
and hence
$\tL \cocp \condi_\Theta .$
If $([\E_i ])_{i \in I}$ is decreasing,
we have $\tL \cocp \Lambda_{\E_i} \cocp \condi_\Theta$
for any $i \in I .$
Thus by Lemma~\ref{lemm:ech}
$\tL = \Lambda_{\tE} $ for some 
$\tE \in \exper (\Theta ) .$
Then $[\tE]$ is a supremum (respectively, infimum)
of $([\E_i ])_{i \in I} .$
If $[\E_i] \in \Ecl (\Theta)$ for all $i \in I ,$ 
$\Lambda_{\E_i} \in \CHqc (\calL (\ltwo (\Theta)))$  for all $i \in I$
and hence 
$\tL =\Lambda_{\tE} \in  \CHqc (\calL (\ltwo (\Theta))) $ 
by Theorem~\ref{theo:mainch}~(ii).
In this case $[\tE] \in \Ecl (\Theta) .$  
\qed

The constructions of the supremum and the infimum in Lemmas~\ref{lemm:inch}
and \ref{lemm:decreasingch} are summarized as in the following corollary.

\begin{corollary}
\label{coro:const}
\begin{enumerate}[1]
\item
Let $\Min$ be a von Neumann algebra.
\begin{enumerate}[(i)]
\item
For any increasing transfinite sequence 
$([\Lambda_\alpha])_{\alpha < \alpha_0}$ in $\CH (\Min) $
and the supremum 
$[\tL] := \sup_{\alpha < \alpha_0} [\Lambda_\alpha]$ in $\CH (\Min) ,$
the representative elements $\Lambda_\alpha \in \ncpchset{\M_\alpha}{\Min}$ 
and $\tL \in \ncpchset{\tM}{\Min} $ can be taken such that 
$(\M_\alpha)_{\alpha < \alpha_0}$ is an increasing transfinite sequence of
von Neumann subalgebras of $\tM ,$
$\bigcup_{\alpha < \alpha_0} \M_\alpha $ is ultraweakly dense in $\tM ,$
and $\Lambda_\alpha$ is the restriction of $\tL$ to $\M_\alpha .$
\item
For any decreasing transfinite sequence 
$([\Lambda_\alpha])_{\alpha < \alpha_0}$ in $\CH (\Min) $
and the infimum 
$[\tL] := \inf_{\alpha < \alpha_0} [\Lambda_\alpha]$ in $\CH (\Min) ,$
the representative elements $\Lambda_\alpha \in \ncpchset{\M_\alpha}{\Min}$ 
and $\tL \in \ncpchset{\tM}{\Min} $ can be taken such that 
$(\M_\alpha)_{\alpha < \alpha_0}$ is a decreasing transfinite sequence of
von Neumann subalgebras of $\M_0 ,$
$\bigcap_{\alpha < \alpha_0} \M_\alpha = \tM ,$
$\Lambda_\alpha$ is the restriction of $\Lambda_0$ to $\M_\alpha  ,$
and 
$\tL$ is the restriction of $\Lambda_0$ (and hence of $\Lambda_\alpha$)
to $\tM .$
\end{enumerate}
\item
Let $\Theta \nonempty$ be a set.
\begin{enumerate}[(i)]
\item
For any increasing transfinite sequence 
$([\E_\alpha])_{\alpha < \alpha_0}$ in $\ETh $
and the supremum 
$[\tE] := \sup_{\alpha < \alpha_0} [\E_\alpha]$ in $\ETh ,$
the representative elements 
$\E_\alpha = (\M_\alpha , \phth^\alpha : \thin)$ 
and $\tE = (\tM , \widetilde{\vph}_\theta : \thin)$ can be taken such that 
$(\M_\alpha)_{\alpha < \alpha_0}$ is an increasing transfinite sequence of
von Neumann subalgebras of $\tM ,$
$\bigcup_{\alpha < \alpha_0} \M_\alpha $ is ultraweakly dense in $\tM ,$
and $\phth^\alpha$ is the restriction of $\widetilde{\vph}_\theta$ 
to $\M_\alpha .$
\item
For any decreasing transfinite sequence 
$([\E_\alpha])_{\alpha < \alpha_0}$ in $\ETh $
and the infimum 
$[\tE] := \inf_{\alpha < \alpha_0} [\E_\alpha]$ in $\ETh ,$
the representative elements 
$\E_\alpha = (\M_\alpha , \phth^\alpha : \thin)$ 
and $\tE = (\tM , \widetilde{\vph}_\theta : \thin)$ can be taken such that 
$(\M_\alpha)_{\alpha < \alpha_0}$ is a decreasing transfinite sequence of
von Neumann subalgebras of $\M_0 ,$
$\bigcap_{\alpha < \alpha_0} \M_\alpha = \tM ,$
$\phth^\alpha$ is the restriction of $\phth^0$ to $\M_\alpha  ,$
and 
$\widetilde{\vph}_\theta$ is the restriction of 
$\phth^0$ (and hence of $\phth^\alpha$)
to $\tM .$
\end{enumerate}
\end{enumerate}
\end{corollary}
\begin{proof}
We first show the claim~1~(i) when $\Min = \LHin .$
By applying the construction of Lemma~\ref{lemm:inch}, 
there exists a von Neumann algebra 
$\M ,$ a normal channel $\Lambda \in \ncpchset{\M}{\LHin} ,$
and an increasing transfinite sequence $(\A_\alpha)_{\alpha < \alpha_0}$
of \cstar-subalgebras of $\M$ such that
$\A := \bigcup_{\alpha < \alpha_0} \A_\alpha$ is ultraweakly dense in $\M ,$
$\Lambda_\alpha \eqcp \Lambda \rvert_{\A_\alpha} ,$
and $[\Lambda] = \sup_{\alpha < \alpha_0} [\Lambda_\alpha] .$  
We take a Stinespring representation $(\cK , \pi , V)$ of $\Lambda$
such that $\pi \colon \M \to \LK$ is faithful and normal.
We put $\M_\alpha := \pi (\A_\alpha)^{\prime \prime} ,$
$\tM := \pi (\M) ,$
and define $\tL \in \ncpchset{\tM}{\LHin}$
and redefine $\Gamma_\alpha \in \ncpchset{\M_\alpha}{\LHin}$
by
\begin{gather*}
	\tL (A) := V^\ast A V  \quad (A \in \tM) ,
	\\
	\Gamma_\alpha (B) := V^\ast B V  \quad (B \in \M_\alpha) .
\end{gather*}
Then $\tL \eqcp \Lambda. $
Furthermore, since $\Gamma_\alpha$ is the double conjugate channel 
of $\Lambda_\alpha ,$ the normality of $\Lambda_\alpha$ implies
$\Lambda_\alpha \eqcp \Gamma_\alpha .$ 
Hence if we redefine $\Lambda_\alpha$ as $\Gamma_\alpha ,$
then $(\M_\alpha , \Lambda_\alpha)_{\alpha < \alpha_0},$
$\tM$ and $\tL$ satisfy all the conditions of the claim. 
The claim for general $\Min$ reduces to the case of $\LHin .$

The claim~1~(ii) can be shown similarly by using the construction given  
in the first paragraph of the proof of Lemma~\ref{lemm:decreasingqc}.

The claim~2 follows from the claim~1 by considering $\Lambda_{\E_\alpha}$
and $\Lambda_{\tE} .$
\end{proof}

%\section{Finite-dimensional system}\label{sec:finite}
%For a set $\Theta$ and a von Neumann algebras $\M , \N,$
%let 
%\begin{gather*}
%	\E (\M , \Theta) := 
%	\set{ [\E] \in \ETh |  \text{the outcome space of $\E$ is $\M$} } ,
%	\\
%	\CH (\M , \N) 
%	:=
%	\set{
%	[\Lambda] \in \CH (\N)
%	|
%	\text{the outcome space of $\Lambda$ is $\M$}
%	} .
%\end{gather*}
%In this section, we consider 
%$\E (\LH , \Theta ) $ and $ \CH (\LH , \LH)$ 
%when $\cH$ is finite-dimensional
%and $\Theta \nonempty$ is finite
%and show that $\E (\LH , \Theta ) $ and $ \CH (\LH , \LH)$
%are upper and lower Dedekind-closed subsets of $\E (\Theta)$ and $\CH (\LH) ,$
%respectively.
%
%\subsection{LeCam distance}

\section{Markov process of statistical experiments or channels}\label{sec:markov}
In this section we consider two examples of 
homogeneous Markov processes of statistical experiments or channels
and derive the infima of the channels.
%The finite dimensional case is considered in Ref.~\onlinecite{matsumoto2012loss}.
Before going into the examples,
we first establish the general relation between 
a net of decreasing channels and a Markov process of statistical experiments
induced by the channels.

\begin{definition}
\label{defi:markov}
Let $\Min$ be a von Neumann algebra
and let $([\Lambda_i])_{i\in I}$ be a decreasing net in 
$\CH (\Min) .$
For a statistical experiment
$\E = (\Min , \phth : \thin)$ on $\Min ,$
the net of statistical experiments $\E_i := \Lambda_{i\ast} (\E)$
$(i \in I)$
is called the Markov process induced by $\E$ and $(\Lambda_i)_{i \in I} .$
\end{definition}

In Definition~\ref{defi:markov},
the nets $([\Lambda_i])_{i \in I}$
and
$([\E_i])_{i \in I}$
are decreasing in $\CH (\Min)$ and $\ETh ,$
respectively.
Hence Theorems~\ref{theo:mainch} and \ref{theo:mainex}
imply the existence of the infima
$[\tL] := \inf_{i \in I} [\Lambda_i] \in \CH (\Min)$
and
$[\tE] := \inf_{i \in I} [\E_i] \in \ETh .$

\begin{proposition}
\label{prop:markov}
In the above setting,
$[\tE] = [\tL_\ast (\E)] .$
\end{proposition}
From Proposition~\ref{prop:markov}, we can know the infimum statistical experiment
$\tE$ if we know the infimum channel $\tL .$

Proposition~\ref{prop:markov} is a corollary of the following more general fact.
\begin{proposition}\label{prop:sconti}
Let $\E = (\Min , \phth : \thin)$ be a statistical experiment.
Then the map
\[
	h_{\E}
	\colon
	\CH (\Min)
	\in [\Lambda]
	\longmapsto
	[\Lambda_\ast (\E)]
	\in
	\ETh
\]
is bi-Scott-continuous.
\end{proposition}
For the proof of Proposition~\ref{prop:sconti}, we first show the following lemma, 
which can be regarded as the \lq\lq{}conjugate\rq\rq{} version of Lemma~\ref{lemm:increasing_alg}.
\begin{lemma}
\label{lemm:decreasing_alg}
Let $\E_0 := (\M , \phth : \thin) $ be a statistical experiment,
let $(\M_i)_{i \in I}$ be a net of von Neumann subalgebras of $\M$
decreasing with respect to the set inclusion $\subseteq ,$
let $\phth^i$ be the restriction of $\phth$ to $\M_i ,$
and let $\E_i := (\M_i , \phth^i : \thin) .$
Then $\inf_{i \in I} [\E_i] = [\tE] ,$
where 
$\tE:= (\tM , \widetilde{\vph}_\theta : \thin) ,$
$\tM := \bigcap_{i \in I} \M_i ,$
and 
$\widetilde{\vph}_\theta$ is the restriction of $\phth$ to $\tM .$ 
\end{lemma}
\begin{proof}
That
$[\tE] \cocp [\E_i]$
for all $i \in I$ is immediate from the definition.
Therefore we have only to show 
$\F \cocp \tE $
for any statistical experiment
$\F = \seF$ satisfying $\F \cocp \E_i $
for all $i \in I .$
By assumption there exists a channel
$\alpha_i \in \cpchset{\N}{\M_i} \subseteq \cpchset{\N}{\M}$ such that
$\psth = \phth^i \circ \alpha_i = \phth \circ \alpha_i$
for all $i \in I .$
By Tychonoff's theorem, there exist a subnet $(\alpha_{i(j)})_{j \in J}$
and a channel $\alpha \in \cpchset{\N}{\M}$ such that
$\alpha_{i(j)} (A) \xrightarrow{uw} \alpha(A)$ for each $A \in \N .$
Since $\alpha_{i(j)} (A) \in \M_i$ eventually, 
%for each $A\in \N ,$
the ultraweak closedness of $\M_i$ implies
$\alpha (A) = \uwlim_{j \in J} \alpha_{i(j)} (A) \in \M_i $
for each $A \in \N$ and each $i \in I .$
Hence $\alpha \in \cpchset{\N}{\tM} .$
Furthermore, for each $\thin$ and each $A \in \N ,$
\[
\widetilde{\vph}_\theta \circ \alpha (A)
=
\uwlim_{j \in J} \phth \circ \alpha_{i(j)} (A)
=
\psth (A) ,
\]
where we used the normality of $\phth .$
Therefore $\F \cocp \tE .$
\end{proof}

\noindent
\textit{Proof of Proposition~\ref{prop:sconti}.}
For normal channels $\Lambda , \Gamma \in \ncpchset{}{\Min} ,$
we can easily see that
$\Lambda \cocp \Gamma$ implies $\Lambda_\ast (\E) \cocp \Gamma_\ast (\E) .$
Hence $h_\E$ is well-defined and order-preserving.

We first show the lower Scott-continuity of $h_\E . $
%since the proof for the upper Scott-continuity is almost parallel. 
By Lemma~\ref{lemm:scott} we have only to show
$h_\E  (  \inf_{\alpha < \alpha_0}  [\Lambda_\alpha]  ) 
=  \inf_{\alpha < \alpha_0}  h_\E ( [\Lambda_\alpha]  )  $
for any decreasing transfinite sequence $([\Lambda_\alpha])_{\alpha < \alpha_0}$
in $\CH (\Min) .$
We can take the channels $\Lambda_\alpha \in \ncpchset{\M_\alpha}{\Min}$
and $\tL \in \ncpchset{\tM}{\Min}$
as in 1~(ii) of Corollary~\ref{coro:const}.
Then the claim is immediate from Lemma~\ref{lemm:decreasing_alg}.

The upper Scott-continuity follows similarly from 
Corollary~\ref{coro:const} and 
by applying Lemma~\ref{lemm:increasing_alg} to $\Lambda_{\E_\alpha}$\rq{}s.
\qed

Now we consider two examples of homogeneous Markov processes 
on infinite-dimensional separable Hilbert spaces.
Here the term \lq\lq{}homogeneous\rq\rq{} means that
the family of decreasing channels is a discrete or continuous one-parameter semigroup
of channels with the same input and outcome space.

\subsection{Block-diagonalization with irrational translation}
\label{subsec:incommensurate}
Let $ L^p (\realn ) $ $(1 \leq p \leq \infty)$
be the $L^p$-space of the 
Lebesgue measure on the real line $\realn $
and let $\cH := L^2 (\realn )  .$
We denote by $\id_{L^\infty (\realn)} \in \ncpchset{ L^\infty (\realn) }{ \LH}$ 
the natural embedding,
where each $ f \in L^\infty (\realn)$ is identified with an element of $\LH$ by
\[
	(f g ) (x) := f(x)g (x)
	\quad
	( g \in L^2 (\realn)) .
\]
Operationally, $\id_{L^\infty (\realn)}$ 
is the QC channel corresponding to the
position measurement of a $1$-dimensional quantum particle.
%For a complex-valued Lebesgue-measurable function $f$ on $\realn ,$ 
%we denote by $[f]$ the equivalence class consisting of 
%functions equal to $f$ $\mu$-almost everywhere ($\mu$-a.e.).
For each $a \in \realn ,$ we define a projection $P_a \in L^\infty (\realn)$ 
and a unitary operator $U_a \in \LH$ 
by
\[
	  P_a  := \chi_{[a , a+1)}  ,
	 \quad
	 (U_a f) (x) 
	 :=
	 f(x-a)
	 \quad
	 (x \in \realn, f \in L^2 (\realn)) ,
\]
where $\chi_S$ is the indicator function of a set $S .$
We fix an irrational number $\alpha \in \realn$ and
define a normal channel $\Lshift \in \ncpch{\LH}$
by
\[
	\Lshift (A)
	:=
	\sum_{n\in \integer}
	P_n U_\alpha^\ast A  U_\alpha P_n 
	\quad
	(A \in \LH) ,
\]
where $\integer$ denotes the set of integers.
In the Schr\"odinger picture,
$\Lshift$ corresponds to the block-diagonalization in the position-representation 
followed by an irrational translation.
Then the sequence
$([(\Lshift )^k])_{k \in \natn}$ is decreasing in
$\CH (\LH)$ and 
has an infimum $\inf_{k \in \natn} [(\Lshift )^k]  ,$
which is given by

\begin{theorem}\label{theo:shift}
	$\inf_{k \in \natn} [(\Lshift )^k] = [\id_{L^\infty (\realn)}]
	 = \sup_{k\in \natn}  [((\Lshift )^k)^c]   .$
\end{theorem}
\begin{proof}
From the definitions, we can easily see
$P_n U_\alpha = U_\alpha P_{n-\alpha}$
$(n \in \integer) .$
Hence
\[
	(U_\alpha P_{n_k })
	\cdots 
	(U_\alpha P_{n_2} )
	(U_\alpha P_{n_1 })
	=
	U_{k\alpha} 
	P_{n_k - (k-1) \alpha}
	\cdots 
	P_{n_2 - \alpha} 
	P_{n_1}
\]
for each $k \geq 2 $ and each 
$\{ n_j \}_{j=1}^k \subseteq \integer .$
Therefore we have
\[
	( \Lshift )^k
	=
	\Delta_k \circ \pi_{U_{k\alpha}} ,
\]
where
\begin{gather*}
	\Delta_k (A)
	:=
	\sum_{ n \in \integer}
	\sum_{l =0}^{k-1}
	P_{n,l}^k A P_{n,l}^k ,
	\\
	P_{n,l}^k 
	:=
	\chi_{ [ n+t_{k,l} , n+t_{k,l+1}  )  } \in L^\infty (\realn) , \\
	\pi_U (A)
	:=
	U^\ast A U ,
\end{gather*}
and $0 = t_{k,0} < t_{k,1} < \cdots < t_{k,k-1} < t_{k,k} =1$ is the division of 
the unit interval $[0,1]$ such that
the set $\{ t_{k,l} \}_{l=1}^{k-1}$ coincides with the set of 
the fractional parts of $- \alpha , -2 \alpha , \dots , - (k-1)\alpha .$
Since the channel $\pi_{U_{k\alpha}}$ has the inverse $\pi_{U_{-k\alpha}} ,$
we have $(\Lshift)^k \eqcp \Delta_k .$ 
We can easily see that $\Delta_k $ 
is randomization-equivalent to 
the identity channel
$\id_{\M_k} \in \ncpchset{ \M_k}{\LH },$
where $\M_k$ is the von Neumann algebra on $\cH$ defined by
\[
	\M_k := \bigoplus_{n \in \integer , 0 \leq l \leq k-1}
	\calL (P_{n,l}^k \cH ) .
\]
Therefore the conjugate channels $( (\Lshift)^n )^c \eqcp \Delta_n^c$
are randomization-equivalent to the identity channel 
$\id_{\M_k^\prime} \in \ncpchset{\M_k^\prime}{\LH}  ,$
where $\M_k^\prime$ is the commutant of $\M_k$ and given by
\[
	\M_k^\prime
	=
	\bigoplus_{n \in \integer , 0 \leq l \leq k-1}
	\cmplx P_{n,l}^k .
\]
Hence we have only to show
\begin{equation}
	\inf_{k \in \natn } [\id_{\M_k}]
	=
	[\id_{L^\infty (\realn)}]
	=
	\sup_{k \in \natn} 
	[\id_{\M_k^\prime}] .
	\label{eq:tobeshown}
\end{equation}
From the irrationality of $\alpha ,$
the width of the division
$\max_{0 \leq l \leq k-1} (t_{k,l +1} - t_{k, l})$
converges to $0$ when $k \to \infty  .$
This implies that there exists a sequence $(Q_n)_{n \in \natn}$ of projections 
in $\bigcup_{k \in \natn} \M_k^\prime$
weakly convergent to $\chi_I \in L^\infty (\realn)$
for any interval $I .$
Therefore $\A_0 := \bigcup_{k \in \natn} \M_k^\prime$ is a weakly dense 
$\ast$-subalgebra of $L^\infty (\realn) $
and hence Lemma~\ref{lemm:increasing_alg} implies 
\[
\sup_{k \in \natn} [\id_{\M_k^\prime}]
= [  (\id_{\A}^c)^c  ]
= [  \id_{\A^{\prime \prime}}  ]
=[  \id_{L^\infty (\realn)}  ] ,
\]
where $\A$ is the norm closure of $\A_0 .$
From this and $L^\infty (\realn)^\prime =L^\infty (\realn) ,$ we obtain
$\inf_{k \in \natn } [\id_{\M_k}]
	=
	[(\id_{L^\infty (\realn)} )^c] 
	=
	[\id_{L^\infty (\realn)^\prime}]
	=
	 [\id_{L^\infty (\realn)} ] ,
	$
proving \eqref{eq:tobeshown}.
\end{proof}
From Proposition~\ref{prop:markov} and Theorem~\ref{theo:shift},
we can see that if we perform the channel $\Lshift$ on the quantum 
system corresponding to $\cH $ infinitely many times, 
the information remaining in the system is exactly that obtained 
from the position measurement.

Theorem~\ref{theo:shift} also shows that
the set
\[
	\ncpch{\LH} / \eqcp \, =
	\set{[\Lambda] \in \CH (\LH) | \text{the outcome space of $\Lambda$ is $\LH$}   }
\]
is neither upper nor lower Dedekind-closed subset of $\CH (\LH) $
since 
\begin{equation}
[\id_{L^\infty (\realn)}] \not\in \ncpch{\LH} / \eqcp .
\label{eq:notin}
\end{equation}
The claim~\eqref{eq:notin} follows from 
the minimal sufficiency~\cite{kuramochi2017minimal} 
of $\id_{L^\infty (\realn)} ,$
which is immediate from the injectivity of $\id_{L^\infty (\realn)} ,$
and from that
any channel $\Lambda \in \ncpch{\LH} $ is randomization-equivalent to 
a unique, up to normal isomorphism, minimal sufficient channel with 
a \emph{discrete} type~I outcome algebra.
The latter fact follows by applying the discussion in
Ref.~\onlinecite{kuramochi2018access}
(Section~3.3) to $\E_\Lambda .$
Such non-closedness is peculiar to infinite-dimensional Hilbert spaces since 
for any finite-dimensional $\cH$ the set $\ncpch{\LH} / \eqcp$ is 
upper and lower Dedekind-closed,
which can be shown
by using the compactness 
of $\ncpch{\LH} = \cpch{\LH} $
in the norm topology.

\subsection{Ideal quantum linear amplifier}
\label{subsec:amplifier}
In this subsection, 
we consider the ideal quantum linear 
amplifier.~\cite{PhysRevA.86.063802}
Let $\cH = \ltwo (\natz) ,$
where $\natz$ is the set of natural numbers containing $0.$
We identify $\cH$ with the system of $1$-mode photon field,
in which the orthonormal basis $(\delta_n)_{n \in \natz}$ is the set of
eigenvectors of the unbounded photon number operator
$
	\sum_{n \in \natz} n \ket{\delta_n} \bra{\delta_n} .
$
For each $t >0,$
we define the ideal quantum linear amplifier channel 
$\Lamp_t \in \ncpch{\LH}$ 
by the following Kraus sum form:
\begin{gather}
	\Lamp_t (A)
	:=
	\sum_{m \in \natz} 
	M_m (t)^\ast A M_m (t) ,
	\notag %\label{eq:Lamp}
	\\
	M_m (t)
	:=
	\sum_{n \in \natz}
	\left[
	(1-e^{-t})^m e^{- (n+1)t}
	\binom{n+m}{m}
	\right]^{1/2}
	\ket{\delta_{n+m}}
	\bra{\delta_n}  ,
	\notag
\end{gather}
where the RHSs of the first and second 
equalities are convergent in the weak and norm topologies, 
respectively.
If we define the creation and annihilation operators
\[
	a^\ast
	:=
	\sum_{n \in \natz}
	\sqrt{n+1} \ket{\delta_{n+1}} \bra{\delta_{n}} ,
	\quad
	a :=
	\sum_{n \in \natz}
	\sqrt{n+1} 
	\ket{\delta_{n}} \bra{\delta_{n+1}} ,
\]
the operator 
$M_m (t)$ can be written as
\begin{equation}
	M_m (t)
	=
	\left[
	\frac{( e^t -1 )^m }{m!}
	\right]^{1/2}
	e^{-t a a^\ast /2} 
	(a^\ast)^m.
	\label{eq:Mmt}
\end{equation}
In the Schr\"odinger picture, 
the corresponding predual map is given by
\begin{equation}
	(\Lamp_t)_\ast
	(T)
	=
	\sum_{m \in \natz}
	M_m (t) T M_m(t)^\ast 
	\quad
	(T \in \mathcal{T} (\cH)) ,
	\label{eq:Lamps}
\end{equation}
where $\mathcal{T} (\cH)$ denotes the set of trace-class operators 
on $\cH$ and the RHS of \eqref{eq:Lamps} is convergent in the trace-norm 
topology.
From \eqref{eq:Mmt}, we can see that 
the RHS of \eqref{eq:Lamps}
coincides with the expression given in 
Ref.~\onlinecite{PhysRevLett.68.3424} (Eq.~(21)).

%(Lindblad equation)

The channels $(\Lamp_t)_{t >0}$ satisfy the semigroup property
$\Lamp_{t} \circ \Lamp_s = \Lamp_{t+s}$
$(t,s >0 ) .$
This can be shown
by comparing the Q-functions~\cite{1940264,PhysRev.138.B274}
of states as follows.
For each normal state $\vph \in \Ss (\LH)$
we define the Q-function of $\vph$ by
\[
	Q_\vph (\alpha)
	:=
	\pi^{-1} \vph (\ket{\psi_\alpha} \bra{\psi_\alpha}) 
	\quad
	(\alpha \in \cmplx) 
	,
\]
where
\[
	\psi_\alpha
	:=
	e^{- |\alpha|^2 /2}
	\sum_{n \in \natz}
	\frac{\alpha^n}{\sqrt{n!}} 
	\delta_n
\]
is the coherent state.~\cite{PhysRev.131.2766}
Coherent states satisfy the overcompleteness relation
\[
	\pi^{-1}
	\int_{\cmplx}
	\ket{\psi_\alpha}
	\bra{\psi_\alpha}
	d^2 \alpha 
	=
	\unit_{\cH} ,
\]
where $d^2 \alpha = d (\Real \alpha) d (\Imag \alpha)$
is the $2$-dimensional Lebesgue measure on $\cmplx$
and the integral is in the weak sense.
The positive-operator valued measure 
$\pi^{-1} \ket{\psi_\alpha} \bra{\psi_\alpha} d^2 \alpha$
defined on the Borel $\sigma$-algebra of $\cmplx$ is
called the Bargmann measure~\cite{bargmann1961,klauder1968fundamentals}
and the Q-function $Q_\vph$ is the outcome probability density 
when the input state is $\vph .$
Hence
$Q_\vph \in L^1(\cmplx) \cap L^\infty (\cmplx) $
for each $\vph \in \Ss (\LH) ,$
where $L^p (\cmplx)$ is the $L^p$-space for the Lebesgue measure
$d^2 \alpha$ on $\cmplx .$
The QC channel corresponding to the Bargmann measure
in the Heisenberg picture is the normal channel
$\gbarg \in \ncpchset{L^\infty (\cmplx)}{\LH}$
given by
\[
	\gbarg (f)
	:=
	\pi^{-1}
	\int_\cmplx f (\alpha) \ket{\psi_\alpha}
	\bra{\psi_\alpha}
	d^2 \alpha
	\quad
	(f \in L^\infty (\cmplx) ) .
\]
It is known~\cite{PhysRev.138.B274,klauder1968fundamentals} that 
the map
\begin{equation}
	\Ss (\LH) \ni \vph
	\longmapsto 
	Q_\vph \in L^1 (\cmplx)
	\label{eq:Qinj}
\end{equation}
is injective. This fact is equivalent to the informational completeness 
of the Bargmann measure.

Now we show the semigroup property of $(\Lamp_t)_{t>0} .$
Since
\[
	M_m^\ast (t) \psi_\alpha 
	=
	\exp \left[
	- \frac{|\alpha|^2}{2} (1-e^{-t}) 
	-\frac{t}{2}
	\right]
	\frac{\alpha^m (1-e^{-t})^{m/2}}{\sqrt{m!}}
	\psi_{e^{-t/2}\alpha} ,
\]
we have
\begin{align}
	Q_{\vph \circ \Lamp_t } (\alpha)
	&=
	\pi^{-1}
	\sum_{m \in \natz}
	\vph
	(
	\ket{M_m (t)^\ast \psi_\alpha}
	\bra{M_m (t)^\ast \psi_\alpha} 
	) 
	\notag
	\\
	&=
	\pi^{-1}
	\exp \left[
	- |\alpha|^2 (1-e^{-t}) 
	-t
	\right]
	\sum_{m \in \natz}
	\frac{|\alpha|^{2m} (1-e^{-t})^m}{m!}
	\vph 
	(\ket{\psi_{e^{-t/2}\alpha}} 
	\bra{\psi_{e^{-t/2}\alpha}} )
	\notag
	\\
	&=
	e^{-t}
	Q_{\vph} (e^{-t/2}\alpha) .
	\label{eq:Qamp}
\end{align}
Hence for each $\vph \in \Ss (\LH),$ each $s,t \in (0,\infty) ,$
and each $\alpha \in \cmplx ,$ we have
\[
	Q_{\vph \circ \Lamp_{s+t}} (\alpha)
	=
	e^{-s-t} Q_\vph (e^{- (s+t)/2}\alpha)
	=
	e^{-t}
	Q_{\vph \circ \Lamp_s} 
	(e^{-t/2} \alpha)
	=
	Q_{\vph \circ \Lamp_s \circ \Lamp_t} (\alpha).
\]
From the injectivity of \eqref{eq:Qinj},
this implies the semigroup property $\Lamp_{s+t} = \Lamp_s \circ \Lamp_t .$

Equation~\eqref{eq:Qamp} shows that the Q-function 
of the outcome state $\vph \circ \Lamp_t$
is that of the input state $\vph$
amplified by the factor $e^{t/2}.$
Such a channel is used to amplify small input signals in quantum optics.
For further information on the physical significance of this channel, 
see Ref.~\onlinecite{PhysRevA.86.063802} and references therein.

By the semigroup property, 
$([\Lamp_t])_{t>0}$ is decreasing in $\CH (\LH)$ and has an infimum,
which is given by
\begin{theorem}
\label{theo:amp}
$
\inf_{t>0} [\Lamp_t]
=
[(\gbarg)^c]
=
[\gbarg]
=
\sup_{t>0} [( \Lamp_t )^c] .
$
\end{theorem}
\begin{proof}
The first equality of the claim is immediate from the third one.
The second equality follows from that
the minimal Stinespring representation of $\gbarg$ is non-degenerate
(Ref.~\onlinecite{ISI:000236487200002}, Section~13).
Thus we have only to prove the third equality.
For this we first derive an explicit expression of the conjugate channel
$(\Lamp_t)^c $
for $t>0 .$
We define a linear isometry 
$V_t \colon \cH \to \cH \otimes \cH$ by
\[
	V_t x
	:=
	\sum_{m \in \natz}
	M_m (t) x \otimes \delta_m
	\quad
	(x \in \cH) .
\]
Then $\Lamp_t$ has the Stinespring representation
$\Lamp_t (A) = V_t^\ast (A \otimes \unit_\cH) V_t$
$(A \in \LH) .$
Hence we may define the conjugate channel
$(\Lamp_t)^c \in \ncpch{\LH}$
by
\[
(\Lamp_t)^c (A) := V_t^\ast (\unit_\cH \otimes A) V_t 
\quad 
(A \in \LH) .
\]
We show that
\begin{equation}
	(\Lamp_t)^c (A)
	=
	\pi^{-1}
	\int_\cmplx
	\braket{\psi_{c_t \overline{\alpha}}| A \psi_{c_t \overline{\alpha}}}
	\ket{\psi_\alpha}
	\bra{\psi_\alpha}
	d^2 \alpha 
	\quad
	(A \in \LH) ,
	\label{eq:lampc}
\end{equation}
where $c_t := (e^t -1)^{1/2} .$
We write the RHS of \eqref{eq:lampc} as
$\Gamma_t (A) .$
Then we can see that $\Gamma_t$ is a well-defined normal channel
in $\ncpch{\LH} .$
Hence, to prove \eqref{eq:lampc},
it is sufficient to show
\begin{equation}
	\braket{\delta_m |(\Lamp_t)^c (\ket{\delta_k}\bra{\delta_l}) \delta_n}
	=
	\braket{\delta_m | \Gamma_t (\ket{\delta_k}\bra{\delta_l}) \delta_n}
	\label{eq:lampc2}
\end{equation}
for each $k,l,m,n \in \natz .$
The LHS of \eqref{eq:lampc2} is evaluated as
\begin{align}
	\braket{\delta_m |(\Lamp_t)^c (\ket{\delta_k}\bra{\delta_l}) \delta_n}
	&=
	\braket{\delta_m | M_k(t)^\ast M_l(t) \delta_n}
	\notag 
	\\
	&=
	\delta_{n+l , k+m}
	(e^t -1)^{\frac{l+k}{2}} e^{- (n+l+1)t}
	\left[
	\binom{n+l}{l} \binom{k+m}{k}
	\right]^{1/2} , 
	\label{eq:kekka}
\end{align}
where $\delta_{a,b}$ denotes the Kronecker delta.
On the other hand, the RHS of \eqref{eq:lampc2} is given by
\begin{align*}
	&\braket{\delta_m | \Gamma_t (\ket{\delta_k}\bra{\delta_l}) \delta_n}
	\\
	&=
	\pi^{-1}
	\int_\cmplx 
	\braket{\psi_{c_t \overline{\alpha}}| \delta_k }
	\braket{\delta_l |  \psi_{c_t \overline{\alpha}} }
	\braket{\delta_m | \psi_\alpha}
	\braket{\psi_\alpha | \delta_n} d^2 \alpha
	\\
	&=
	\frac{\pi^{-1} c_t^{l+k}}{ \sqrt{k! l! m! n!}  }
	\int_\cmplx
	e^{-(1 + c_t^2) |\alpha|^2 }
	\overline{\alpha}^{l+n} 
	\alpha^{k+m}
	d^2 \alpha
	\\
	&=
	\frac{\pi^{-1} c_t^{l+k}}{ \sqrt{k! l! m! n!}  }
	\int_0^\infty 
	\int_0^{2\pi}
	e^{-(1 + c_t^2) r^2 }
	r^{l+n+k+m+1}
	e^{i(-l-n+k+m) \theta}
	d\theta dr
	\quad 
	(r = |\alpha| , \theta = \arg  \alpha )  
	\\
	&=
	\delta_{n+l , k+m}
	c_t^{l+k} (1+c_t^2)^{-n-l-1}
	\left[
	\binom{n+l}{l} \binom{k+m}{k}
	\right]^{1/2} ,
\end{align*}
which coincides with \eqref{eq:kekka}.
Thus we have shown \eqref{eq:lampc}.

From \eqref{eq:lampc}, we have 
$(\Lamp_t)^c = \gbarg \circ \Phi_t  \cocp \gbarg ,$
where $\Phi_t \in \ncpchset{\LH}{L^\infty (\cmplx)}$ is given by
$
(\Phi_t (A))(\alpha) := 
\braket{\psi_{c_t \overline{\alpha}}| A \psi_{c_t \overline{\alpha}}} .
$
Hence $\sup_{t > 0} [(\Lamp_t)^c] \cocp [\gbarg] .$

We define a channel $\Psi_t \in \ncpchset{L^\infty (\cmplx)}{\LH} $
by
\[
	\Psi_t (f)
	:=
	\pi^{-1} c_t^2
	\int_\cmplx
	f(\alpha)
	\ket{\psi_{c_t \overline{\alpha}}}
	\bra{\psi_{c_t \overline{\alpha}}}
	d^2 \alpha
	\quad
	(f \in L^\infty (\cmplx)) .
\]
We prove
\begin{equation}
	\| \vph \circ( \Lamp_t)^c \circ \Psi_t
	- \vph \circ \gbarg  \|
	\xrightarrow{t \to \infty} 0
	\label{eq:tozero}
\end{equation}
for each $\vph \in \Ss (\LH) .$
We identify, as usual, each element $\psi$ of the predual space 
(i.e.\ the set of ultraweakly continuous linear functionals)
$L^\infty (\cmplx)_\ast$
of $L^\infty (\cmplx)$
with an integrable function $g_\psi \in L^1 (\cmplx)$ by
\[
	\psi (f)
	=
	\int_\cmplx
	f(\alpha ) g_\psi (\alpha) d^2 \alpha
	\quad
	(f \in L^\infty (\cmplx) ) .
\]
Then we have $\| \psi \| = \|g_\psi \|_1 := \int_\cmplx | g(\alpha)| d^2 \alpha$
for each $\psi \in L^\infty (\cmplx)_\ast . $ 
We can also check that $g_{\vph \circ \gbarg} = Q_{\vph}$
for each $\vph \in \Ss (\LH) .$
From 
\begin{align*}
	(\Phi_t \circ \Psi_t (f))(\alpha)
	&=
	\pi^{-1} c_t^2
	\int_\cmplx
	f(\beta)
	\left|
	\braket{ \psi_{c_t \overline{\alpha}}  |
				\psi_{c_t \overline{\beta}}}
	\right|^2
	d^2 \beta
	\\
	&=
	\pi^{-1} c_t^2
	\int_\cmplx 
	e^{-c_t^2 |\alpha - \beta|^2} 
	f(\beta)
	d^2 \beta
	\quad
	(f \in L^\infty (\cmplx)) ,
\end{align*}
we have
\[
	g_{\psi \circ \Phi_t \circ \Psi_t}
	(\beta)
	=
	\pi^{-1} c_t^2
	\int_\cmplx
	e^{-c_t^2 |\alpha - \beta|^2} 
	g_\psi (\alpha)
	d^2 \alpha
	\quad
	(\psi \in L^\infty (\cmplx)_\ast).
\]
Hence for each $\vph \in \Ss (\LH) $ and each $\delta>0 ,$
\begin{align*}
	& \| \vph \circ( \Lamp_t)^c \circ \Psi_t
	- \vph \circ \gbarg  \|
	\\
%	&=
%	\| 
%	g_{\vph \circ( \Lamp_t)^c \circ \Psi_t}
%	-
%	g_{\vph \circ \gbarg}
%	\|_{1}
%	\\
	&=
	\| 
	g_{\vph \circ \gbarg \circ \Phi_t \circ \Psi_t }
	-
	Q_\vph 
	\|_1
	\\
	&=
	\int_\cmplx
	\left|
	\pi^{-1} c_t^2
	\int_\cmplx
	e^{-c_t^2 |\alpha - \beta|^2} 
	Q_\vph (\beta)
	d^2 \beta
	-
	Q_\vph (\alpha)
	\right| 
	d^2 \alpha
	\\
	&\leq
	\pi^{-1} c_t^2
	\int_\cmplx \int_\cmplx
	e^{-c_t^2 |\alpha - \beta|^2} 
	\left| 
	Q_\vph (\alpha) - Q_\vph (\beta)
	\right|
	d^2\beta d^2 \alpha
	\\
	&=
	\pi^{-1} c_t^2
	\int_\cmplx \int_\cmplx
	e^{-c_t^2 |\gamma|^2} 
	\left|
	Q_\vph (\alpha) - Q_\vph (\alpha - \gamma)
	\right|
	d^2\gamma d^2 \alpha
	\\
	&=
	\pi^{-1} c_t^2
	\int_\cmplx \int_\cmplx
	e^{-c_t^2 |\gamma|^2} 
	\left|
	Q_\vph (\alpha) - Q_\vph (\alpha - \gamma)
	\right|
	d^2\alpha d^2 \gamma
	\quad
	(\because \text{Fubini\rq{}s theorem})
	\\
	&\leq
	\pi^{-1} c_t^2
	\int_{\cmplx ,  |\gamma| \leq \delta} 
	e^{-c_t^2 |\gamma|^2}
	\int_\cmplx
	\left|
	Q_\vph (\alpha) - Q_\vph (\alpha - \gamma)
	\right| 
	d^2 \alpha d^2 \gamma
	+ 
	2
	\pi^{-1} c_t^2
	\int_{\cmplx ,  |\gamma| \geq \delta}
	e^{-c_t^2 |\gamma|^2}  d^2 \gamma
	\\
	&\leq
	\sup_{\gamma \in \cmplx , |\gamma | \leq \delta }  
	\|  Q_\vph - T_\gamma Q_\vph  \|_1
	+
		2
	\pi^{-1} c_t^2
	\int_{\cmplx ,  |\gamma| \geq \delta}
	e^{-c_t^2 |\gamma|^2}  d^2 \gamma
	\\
	&\xrightarrow{t\to \infty}
	\sup_{\gamma \in \cmplx , |\gamma | \leq \delta }  
	\|  Q_\vph - T_\gamma Q_\vph  \|_1 ,
\end{align*}
where $T_\gamma$ $(\gamma \in \cmplx)$
is the translation operator on $L^1 (\cmplx)$ defined by
\[
	(T_\gamma g) (\alpha)
	:=
	g(\alpha - \gamma)
	\quad
	(g \in L^1 (\cmplx)) .
\]
From the strong continuity of the translation operator
(e.g.\ Ref.~\onlinecite{folland1999real}, Proposition~8.5),
we have
\[
	\sup_{\gamma \in \cmplx , |\gamma | \leq \delta }  
	\|  Q_\vph - T_\gamma Q_\vph  \|_1 
	\xrightarrow{\delta \to +0} 0,
\]
which implies \eqref{eq:tozero}.

Now we show the third equality of the claim.
We have already shown $\sup_{t>0}[(\Lamp_t)^c]\cocp [\gbarg] .$
To show the converse, take an arbitrary normal channel
$\Gamma \in \ncpchset{\N}{\LH}$
satisfying $ (\Lamp_t)^c \cocp \Gamma$
for all $t > 0 .$
Then for each $t>0$ there exists a channel
$\Lambda_t \in \cpchset{\LH}{\N}$ such that
$(\Lamp_t)^c  = \Gamma \circ \Lambda_t . $
From \eqref{eq:tozero}, we have
\[
	\| \vph \circ \Gamma \circ \Lambda_t \circ \Psi_t
	- \vph \circ \gbarg  \|
	\xrightarrow{t \to \infty} 0
\]
for each $\vph \in \Ss (\LH) ,$
and hence
\[
	\Gamma \circ \Lambda_t \circ \Psi_t (f)
	\xrightarrow{uw}
	 \gbarg (f)
	 \quad
	 (t\to \infty)
\]
for each $f \in L^\infty (\cmplx) .$
By Tychonoff's theorem, there exist a subnet 
$(\Lambda_{t(i)} \circ \Psi_{t(i)})_{i\in I}$
and a channel $\widetilde{\alpha} \in \cpchset{L^\infty (\cmplx)}{\N}$
such that
$\Lambda_{t(i)} \circ \Psi_{t(i)} (f) 
\xrightarrow{uw}
\widetilde{\alpha} (f)
$
for each $f \in L^\infty (\cmplx) .$
Then from the normality of $\Gamma$ we have 
\[
	\Gamma \circ \widetilde{\alpha} (f)
	=
	\uwlim_{i \in I}
	\Gamma \circ \Lambda_{t(i)} \circ \Psi_{t(i)} (f)
	=
	\gbarg (f)
	\quad
	(f \in L^\infty (\cmplx)) ,
\]
which implies $\gbarg = \Gamma \circ \widetilde{\alpha} \cocp \Gamma .$
Therefore 
$\sup_{t > 0 } [(\Lamp_t)^c] = [\gbarg] .$
\end{proof}
Proposition~\ref{prop:markov}
and Theorem~\ref{theo:amp} imply that 
if we infinitely amplify the $1$-mode photon field 
with the channels $(\Lamp_t)_{t>0} ,$
the information remaining in the system is exactly 
that obtained from the quantum measurement
corresponding to the Bargmann measure.

\section{Concluding remarks} \label{sec:final}
In this paper we have investigated randomization-monotone 
nets of quantum statistical experiments or normal channels
and 
established the directed-completeness of 
the sets $\E (\Theta)$ and $\CH (\Min) ,$
together with the Dedekind-closedness of the 
classical statistical experiments $\Ecl (\Theta )$
and QC channels $\CHqc (\Min) .$
In the proof, the concept of channel conjugation 
and the correspondence
\begin{gather*}
	\E (\Theta)
	\ni [\E]
	\longmapsto 
	[\Lambda_\E]
	\in 
	\CH ( \calL (\ltwo (\Theta))   )
%	\\
%	\CH (\Min)
%	\ni [\Lambda]
%	\longmapsto
%	[\E_\Lambda]
%	\in \E (\Ss (\Min))
\end{gather*}
between statistical experiments and channels
have been used to reduce the discussion 
to the case of increasing normal channels.
In the following we list some questions related to our results. 
\begin{enumerate}
\item
In Ref.~\onlinecite{gutajencova2007}, the weak convergence topology 
for statistical experiments is introduced as a generalization of 
the corresponding concept in the classical statistics.~\cite{torgersen1991comparison}
Then, as a natural question, we can ask whether
a randomization-monotone net of statistical experiments 
converges to its supremum or infimum.
\item
As to the weak topology, 
another natural question is whether there is any relation between 
the channel conjugation map
$
	[\Lambda] \mapsto [\Lambda^c]
	\label{eq:conjugation}
$
and the weak convergence topology.
For example, is the conjugation map
$
[\E_{\Lambda}] \mapsto [\E_{\Lambda^c}]
$
continuous in the weak topology?
\item
We can also ask whether the directed-completeness holds for 
other preorder relations for statistical experiments, 
for example that induced by the statistical 
morphisms.~\cite{kaniowski2013quantum,Buscemi2012}
As mentioned in Ref.~\onlinecite{Buscemi2012} (Section~9),
the equivalence by morphism can be defined for 
general probabilistic theories (GPTs)
and it may be natural to consider this problem in the GPT setting. 
\end{enumerate}

\begin{acknowledgements}
This work was supported by the National Natural Science Foundation of China 
(Grants No.~11374375 and No.~11574405).
The author would like to thank Erkka Haapasalo for helpful comments on 
the manuscript.
\end{acknowledgements}
%\bibliography{abbr}   

\end{document}